\documentclass[prd,aps,amsfonts,eqsecnum,superscriptaddress,nofootinbib,notitlepage,longbibliography,final]{revtex4-1}

\usepackage[english]{babel}
\usepackage[utf8]{inputenc}

\usepackage[a4paper, left=2cm, right=2cm, top=2cm, bottom=2cm]{geometry}

\usepackage{amsmath,physics,amsfonts}
\usepackage{graphicx,subcaption}
\usepackage[colorlinks=true, allcolors=blue]{hyperref}
\usepackage{float}
\usepackage{hhline}
\usepackage{mathtools}
\mathtoolsset{showonlyrefs}
\usepackage{soul}
\usepackage{tcolorbox}

\usepackage{pgfplots}
\pgfplotsset{compat=newest}

\usepackage{algorithm}
\usepackage{algpseudocode}

\usepgfplotslibrary{groupplots}

\definecolor{zero}{RGB}{180, 180, 180}
\definecolor{one}{RGB}{206, 104, 104}
\definecolor{two}{RGB}{231, 195, 146}
\definecolor{three}{RGB}{231, 217, 146}
\definecolor{four}{RGB}{209, 223, 141}
\definecolor{five}{RGB}{116, 185, 116}
\definecolor{six}{RGB}{126, 147, 165}
\definecolor{seven}{RGB}{146, 136, 176}
\definecolor{eight}{RGB}{166, 124, 166}

\usepackage{amsthm}

\newtheorem{theorem}{Theorem}[section]
\newtheorem{corollary}{Corollary}[theorem]
\newtheorem{lemma}[theorem]{Lemma}
\newtheorem{prop}[theorem]{Proposition}

\newtheorem{example}[theorem]{Example}
\theoremstyle{definition}
\newtheorem{definition}[theorem]{Definition}
\newtheorem{remark}[theorem]{Remark}

\renewcommand\labelenumi{(\roman{enumi})}
\renewcommand\theenumi\labelenumi

\newcommand{\oscnorm}[1]{{\left\vert\kern-0.25ex\left\vert\kern-0.25ex\left\vert #1 
    \right\vert\kern-0.25ex\right\vert\kern-0.25ex\right\vert}}

\begin{document}

\title{Constant-time equilibration of observables under rapid Lindbladian dynamics}

\author{Štěpán Šmíd}
\email{s (dot) smid23 (at) imperial (dot) ac (dot) uk}
\affiliation{%
Department of Computing, Imperial College London, United Kingdom
}%
\author{Richard Meister}
\affiliation{%
Department of Computing, Imperial College London, United Kingdom
}%
\author{Mario Berta}
\affiliation{%
Institute for Quantum Information, RWTH Aachen University, Germany
}%
\affiliation{%
Department of Computing, Imperial College London, United Kingdom
}%
\author{Roberto Bondesan}
\affiliation{%
Department of Computing, Imperial College London, United Kingdom
}%

\date{\today}


\begin{abstract}
Markovian open-system dynamics have widespread applications throughout quantum information science, including algorithmic state preparation. Their convergence is commonly quantified using the worst case global trace distance between the evolving and stationary states. However, this criterion can be unnecessarily stringent when only physically relevant observables are of interest. Here we introduce and study observable-specific mixing times. We prove that, for quasi-local, rapidly mixing Lindbladians, sums of geometrically local observables equilibrate in a time independent of system size, in contrast to the logarithmic dependence of global state mixing. This separation reduces the runtime of dissipative quantum algorithms, including quantum Gibbs samplers, for estimating quantities such as the Gibbs state energy and local order parameters, yielding an overall scaling that is linear in system size.
Complementing this quantum result, we develop a quantum-inspired classical algorithm for estimating the same quantities. Its runtime is likewise linear in system size, but scaling exponentially in $\mathcal{O}\big(\log(1/\epsilon)^D\big)$, where $D$ denotes the spatial dimension of the lattice.
We further analyse non-interacting Lindbladians over qudits, fermions, and bosons, demonstrating that locality of observables is not always necessary for a qualitatively faster mixing. Small-scale simulations of quantum Gibbs samplers reveal no large hidden constants in our asymptotic analysis and show that the theoretical predictions closely capture the finite-size dynamics. 
\end{abstract}

\maketitle

\section{Overview}

\paragraph*{Open quantum systems and dissipative state preparation.}
Time evolution of a quantum system is the most fundamental process in quantum mechanics. For a closed system, this evolution is invertible and generated by the system's Hamiltonian. But for open quantum systems, one allows the system to also interact with an external heat bath, enabling the dissipation of information from the system into the environment. This interaction is commonly assumed to be weak and memoryless, which is known as the Markovian setting. The study of such evolutions dates back to Lindblad \cite{lindblad1976generators} and Davies \cite{davies1974markovian} in the context of quantum thermalisation. Crucially, unlike unitary evolution, dissipative evolution generated by a Lindbladian $\mathcal{L}$ can converge towards a steady state. We refer the reader to \cite{Mozgunov2020completelypositive, PhysRevB.102.115109, scandi2026thermalization} for modern explanations of quantum master equations. The study of open quantum systems has come a long way, and they have been in recent years widely considered as a powerful tool in the context of algorithmic state preparation \cite{lin2025dissipative}. Seminal works on quantum Gibbs sampling \cite{chen2025efficient,ding2025efficient, gilyen2026quantum} and ground state preparation \cite{PhysRevResearch.6.033147, ding2025end} allow one to construct an artificial open system dynamics converging towards desired steady states and efficiently simulate it on a quantum computer. These often directly simulate an exactly detailed-balanced Lindbladian evolution, but some may also invoke a simplified approximate scheme \cite{ding2025end, fang2026quantum}. In either case, the complexity of the state preparation will depend on the convergence time of the evolution, commonly quantified using the worst case global trace distance from the steady state.
In this work, we will be mainly concerned with the following, practically-motivated, question: \begin{align}
    \textit{Do we need full state convergence in order to measure physically-relevant observables?}
\end{align} We will show that the answer in many cases is actually negative, and that relevant properties can mix qualitatively faster than the full state.\\

\enlargethispage{1cm}
\paragraph*{Main results.}
To study this notion of mixing, we start by defining the mixing time for an observable $O$ (Definition \ref{def:mixing time of observables}) as
\begin{align}
        t_{\textup{mix}}^{(O)}(\epsilon) &= \inf\left\{ t\geq 0 \left|\  \|e^{t\mathcal{L}}[O] - \Tr(O \cdot \sigma) \cdot I\| \leq\right. \epsilon \cdot \|O\| \right\}\,,
\end{align} where $\sigma$ is the unique steady state of the dynamics. Crucially, this notion is compatible with the global state mixing time as $t_\textup{mix}(\epsilon) = \sup_{O} t^{(O)}_\textup{mix}(\epsilon)$, and gives the same guaranteed upper bound on the error when measuring the expectation value of $O$ as $t_\textup{mix}(\epsilon)$, but still leaves the option for specific observables to mix faster than the full state. Our main result (Theorem \ref{thm: constant mixing of geometrically-local observables}) shows that there can often be a qualitative difference between the mixing of the states and that of observables: \begin{align}
    &\textit{For quasi-local and rapidly mixing Lindbladians, observables expressible as sums of geometrically-local}\\ 
    &\hspace{2.6cm} \textit{terms mix in a system-size-independent time, $t_\textup{mix}^{(O)}(\epsilon) =\mathcal{O}(\log(1/\epsilon))$.}
\end{align}
This contrasts the (poly-)logarithmic mixing time needed for a full state convergence under a rapidly mixing Lindbladian. For dissipative algorithms for state preparation, like quantum Gibbs samplers, this has immediate consequences on the end-to-end complexity of evaluating many physically-important properties of the systems, like the Gibbs state energy and local order parameters.
Theorem \ref{thm: constant mixing of geometrically-local observables} further allows us to decrease the complexity of simulating the Lindbladian for both quantum and classical approaches. Separating the evolution into local parts lowers the overall quantum complexity of estimating such expectation values to $\mathcal{O} \left(n\cdot \operatorname{poly}(1/\epsilon)\right)$. Similarly, we also devise a corresponding classical Algorithm~\ref{alg: Estimating expectation values} which estimates these expectation values in time $\mathcal{O} \left(n\cdot e^{\mathcal{O}\left(\log(1/\epsilon)^D\right)}\right)$, where $D$ is the lattice dimension. The quantum simulation then still provides a superpolynomial speed-up with respect to the desired error $\epsilon$ for any $D\geq 2$, whereas in 1D the classical simulation provides a fully efficient algorithm for any short-range Hamiltonian.

To develop deeper understanding of what makes some observables mix qualitatively faster than others, we further study several classes of non-interacting Lindbladians, including separable qudit Lindbladians, and quadratic fermionic and bosonic Lindbladians. For the separable qudits, we provide several large classes of observables mixing in constant time, as well as an example of a simple observable mixing in $\Theta(\log(n/\epsilon))$ time (Example \ref{example: observable mixing in log time}). For fermionic and bosonic Lindbladians, we show that quadratic quantum Gibbs samplers thermalise quadratic observables in a constant time. 
In all three cases, these classes include non-local observables, highlighting that locality is not strictly necessary for mixing times uniform in the system size and that one should study further structure of the observables.

Finally, we complement our theory with numerical simulations of quantum Gibbs samplers, based on both exact diagonalisation and tensor networks, comparing these different notions of mixing times explicitly for the Fermi-Hubbard model and the transverse-field Ising model. Here we see that the qualitatively different scaling can make significant practical differences even for small lattice sizes.
\\

\paragraph*{Proof ideas.}
The main idea is based on a scaling argument combining rapid mixing together with Lieb-Robinson bounds for the Lindbladian evolution itself. While commonly studied for Hamiltonian time evolution \cite{hastings2010localityquantumsystems,Chen2023Speed}, Lieb-Robinson bounds can be often extended to open system dynamics \cite{Poulin2010, Barthel2012}, including that of long-range Lindbladians \cite{Nachtergaele2011, Sweke_2019}. Similarly to the unitary evolution, the bounds on the commutators allow one to localise such dynamics \cite{Barthel2012, cubitt2015stability, SIGAL2026170575}. 
We use this to study the contractivity of a local operator $O_A$ initially supported on a finite local region $A$ by splitting it into two contributions, one from a finite-sized neighbourhood of $A$ (which grows in time), and the other one from the potentially-infinite outside region, \begin{align}
    \left\| e^{t\mathcal{L}}[O_A] - \Tr(e^{t\mathcal{L}}[O_A]) \frac{I}{d^n} \right\| 
    &\leq  \left\| e^{t\mathcal{L}_{B_A(r)}}[O_A] - \Tr_{B_A(r)}(e^{t\mathcal{L}_{B_A(r)}}[O_A]) \frac{I}{d^{|{B_A(r)|}}} \right\| + 2  \left\| e^{t\mathcal{L}}[O_A] - e^{t\mathcal{L}_{B_A(r)}}[O_A]\right\|\\
    &\leq \operatorname{poly}(|B_A(r)|) \cdot e^{-\Delta t} \cdot \|O_A\|+2|A| \cdot J \cdot \frac{e^{vt}-1-vt}{v} \cdot e^{-\gamma r} \cdot \|O_A\|\,.
\end{align} We then make the contribution from the outside region sufficiently small by making the inner region sufficiently large, while the rapid mixing will cause the contribution from the inner region to decay faster than it grows from the increasing size. Importantly, the use of Lieb-Robinson bounds in this manner gets rid of any notion of system size, and hence we get a constant mixing time for the term $O_A$.
Remarkably, this shows that the operator $O_A$ mixes after only interacting with a finite portion of the Lindbladian. 
See Figure \ref{fig:lightcone splitting} for an illustration.\vspace{-0.5cm}
\enlargethispage{1cm}


\begin{figure}[H]
\centering
\begin{tikzpicture}[
    scale=0.58,
    site/.style={circle, fill=zero!85!black, inner sep=0pt, minimum size=2.6pt},
    lat/.style={zero!70, line width=0.3pt},
    every node/.style={font=\small},
  ]

  \def\Rin{0.42}      
  \def\Rr{2.0}        
  \def\Rout{4.45}     
  \def\Lhalf{3}       

  \begin{scope}

    \foreach \i in {0,...,13}{
      \pgfmathsetmacro{\ra}{\Rr + \i*0.11}
      \pgfmathsetmacro{\rb}{\Rr + (\i+1)*0.11}
      \pgfmathsetmacro{\op}{0.85*exp(-0.30*\i)}
      \fill[three, opacity=\op, even odd rule] (0,0) circle (\rb) (0,0) circle (\ra);
    }
    \fill[three, opacity=0.85, even odd rule] (0,0) circle (\Rr) (0,0) circle (\Rin);

    \foreach \i in {-\Lhalf,...,\Lhalf}{
      \draw[lat] (\i,-\Lhalf) -- (\i,\Lhalf);
      \draw[lat] (-\Lhalf,\i) -- (\Lhalf,\i);
    }
    \foreach \x in {-\Lhalf,...,\Lhalf}{\foreach \y in {-\Lhalf,...,\Lhalf}{%
      \node[site] at (\x,\y){};}}

    \fill[one] (0,0) circle (\Rin);
    \node[white] at (0,0) {$A$};

    \draw[-latex, line width=0.5pt, black!80] (62:\Rin) -- (62:\Rr);
    \node[anchor=south east, inner sep=1.5pt] at (62:1.2) {$r = \Theta(t)$};

    \node[anchor=north] at (0,-5.9) {(a)\ \ Inside $B_A(r)$};
  \end{scope}

  \begin{scope}[shift={(11.6,0)}]

    \foreach \i in {0,...,13}{
      \pgfmathsetmacro{\ra}{\Rout + \i*0.11}
      \pgfmathsetmacro{\rb}{\Rout + (\i+1)*0.11}
      \pgfmathsetmacro{\op}{0.85*exp(-0.30*\i)}
      \fill[three, opacity=\op, even odd rule] (0,0) circle (\rb) (0,0) circle (\ra);
    }
    \fill[three, opacity=0.85, even odd rule] (0,0) circle (\Rout) (0,0) circle (\Rr);

    \foreach \i in {-\Lhalf,...,\Lhalf}{
      \draw[lat] (\i,-\Lhalf) -- (\i,\Lhalf);
      \draw[lat] (-\Lhalf,\i) -- (\Lhalf,\i);
    }
    \foreach \x in {-\Lhalf,...,\Lhalf}{\foreach \y in {-\Lhalf,...,\Lhalf}{%
      \node[site] at (\x,\y){};}}

    \foreach \ang in {45,135,225,315}{
      \draw[-latex, line width=0.5pt, black!75] (\ang:1.05) -- (\ang:2.5);
    }

    \draw[black!60, dashed, line width=0.6pt] (0,0) circle (\Rr);
    \fill[one] (0,0) circle (\Rin);
    \node[white] at (0,0) {$A$};

    \node[anchor=north] at (0,-5.9) {(b)\ \ Outside $B_A(r)$};
  \end{scope}

\end{tikzpicture}
\caption{\textbf{Splitting the contractivity of a local term} $\boldsymbol{O_A}$.
\textbf{(a)} Inside the ball $B_A(r)$ the truncated dynamics $e^{t\mathcal{L}_{B_A(r)}}$
equilibrates $O_A$ rapidly, contributing
$\operatorname{poly}(|B_A(r)|)\cdot e^{-\Delta t}\cdot \|O_A\|$; the radius grows linearly with time, but
the exponential decay asymptotically beats the polynomial growth. 
\textbf{(b)} Outside $B_A(r)$ the dynamics is suppressed by the Lieb-Robinson bound, and is made small by the radius increasing with time. 
Neither contribution depends on the system size, so $O_A$ mixes after interacting with only a finite portion of
$\mathcal{L}$.}
\label{fig:lightcone splitting}
\end{figure}
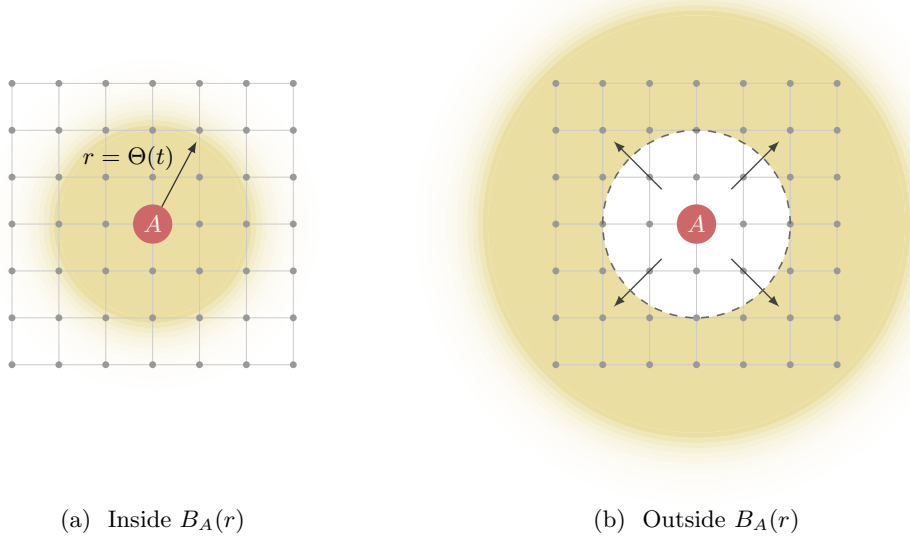


Finally, we lift this result to arbitrary sums of geometrically-local observables $O=\sum_A O_A$ by understanding general behaviour of the observable-specific mixing time for a sum of operators (Lemma \ref{lemma: mixing time of a sum}) and using properties of geometrically-local operators studied in literature on learning quantum processes \cite{Huang2023, lewis2024improved}. In the subsequent study of non-interacting Lindbladians, we proceed by solving the time-evolution of observables in the Heisenberg picture exactly, and, for bosonic systems, further define a regularised notion of the mixing time for unbounded operators (Definition \ref{def:mixing time of bosonic observables}).
\\

\paragraph*{Techniques for bounding mixing times.}
Previous works have generally studied the convergence rates of quantum Markov semigroups through their global mixing times, defined via the trace distance between the evolved state and the steady state. Based on the scaling of this mixing time with respect to the size of the system, we generally distinguish three types of behaviour: \textit{slow mixing} in exponential time, \textit{fast mixing} in polynomial time, and \textit{rapid mixing} in polylogarithmic time. There are several different methods for studying these mixing times, one of the most approachable ones being analysing the spectral gap of the generator. To show fast mixing, it is sufficient to show that the spectral gap closes at most polynomially with respect to the system size. For quantum Gibbs samplers, this has been achieved at high temperatures \cite{rouze2024efficient, bergamaschi2026fast}, for weakly-interacting fermionic systems \cite{smid2025polynomial,tong2025fast}, one-dimensional systems \cite{bergamaschi2026chains}, simplified state preparation algorithms \cite{slezak2026polynomial}, the mean-field Heisenberg model \cite{basso2026spectral}, and for stabiliser Hamiltonians \cite{ding2026polynomial, paez2026efficient}. One can also obtain estimates on the spectral gap by generalising classical methods like hypocoercivity \cite{fang2025mixing} or replica exchange \cite{chen2025quantum}.
In order to show rapid mixing, one requires finer understanding of the dynamics, such as the one achieved with the \textit{modified logarithmic Sobolev inequality} for commuting Hamiltonians \cite{kastoryano2013quantum, capel2020modified, bardet2023rapid, kochanowski2025rapid, stengele2026Abelian, stengele2026modified}, or via the \textit{oscillator norm}, which was recently used to show rapid mixing at high temperatures \cite{Rouze2026}, and for weakly-interacting Gibbs states \cite{smid2025rapid} and ground states \cite{Zhan2026}. Other notable results on establishing rapid mixing also include the Dobrushin condition at high temperatures \cite{bakshi2026dobrushin, bakshi2026rapid}, one-dimensional systems \cite{kim2026rapid}, and self-correcting quantum memories \cite{bergamaschi2026rapid}. All of these works hence bound the convergence in the worst case, but that can be a qualitative overestimate for many useful tasks, such as evaluating the Gibbs state energy and estimating the partition function.\\

\paragraph*{Related works on refined mixing times.}
We note that a similar notion to our observable-specific mixing time was also considered for the ground state preparation Lindbladian in \cite{Zhan2026}, where they've defined an energy-based mixing time. Their definition, however, differed from ours, requiring specific properties of the ground state preparation Lindbladian, and was mainly used as a proxy to the state mixing time for numerical simulations. 
Relevant ideas were also developed in \cite{Kashyap2025}, where the authors theoretically analyse a noisy analogue quantum simulator with tunable damping and weakly-coupled ancillae for geometrically-local Lindbladians. 
They consider the evolution of local observables supported on $\mathcal{O}(1)$ sites, and, while not explicitly studying their mixing time, show that such a simulator can estimate their expectation values in the steady state in a constant time. Moreover, they prove that their protocol is stable under local noise. Their analysis uses similar techniques as we do, combining rapid mixing with Lieb-Robinson bounds on the Lindbladian evolution, but doesn't cover sums of geometrically-local observables with global support, like the Hamiltonian of the system, or Lindbladians with quasi-local interactions. 
Interestingly, assuming that $\text{BQP}\neq\text{BPP}$, they show that no classical algorithm can do this task in time scaling polynomially in the inverse precision for systems in at least two dimensions, agreeing with the complexity of our Algorithm \ref{alg: Estimating expectation values} being quasi-polynomial in $1/\epsilon$ for any $D\geq 2$, while being polynomial in $D=1$.
\\

\paragraph*{Discussion.}
We have introduced a novel way of studying the convergence rate of open quantum systems. This has direct implications for end-to-end complexities of algorithmic tasks, as these require measuring interesting expectation values in the steady state. For rapidly mixing quantum Gibbs samplers \cite{Rouze2026, smid2025rapid, bakshi2026dobrushin, bakshi2026rapid, kim2026rapid} and ground state preparation \cite{Zhan2026}, our approach lowers the computational complexity for tasks such as evaluating the system's energy (Corollaries \ref{cor: constant mixing weakly-interacting Gibbs}, \ref{cor: constant mixing  high-temperature Gibbs}, \ref{cor: constant mixing 1D systems}, \ref{cor: constant mixing ground states}), making these quantum algorithms even more competitive with their classical counterparts \cite{chen2025convergence, mann2026efficient, Alhambra_2021,bakshi2025hightemperaturegibbsstatesunentangled}.
These classical results are based on cluster (or cumulant) expansions and are provably efficient\,---\,i.e.~polynomial in both the system size $n$ and the inverse precision $1/\epsilon$\,---\,strictly only within the radius of convergence of these series.
This occurs at high temperatures or small interaction strengths.
Quantum advantage for evaluating geometrically-local properties, in the form of a superpolynomial speed-up in $1/\epsilon$ dependency with respect to our classical simulation algorithm runtime, is thus expected in the region of temperatures and couplings where the cluster expansion breaks down but rapid mixing persists.
This is also consistent with the hardness results of \cite{Kashyap2025}, which rule out a classical algorithm polynomial in $1/\epsilon$ and the inverse mixing rate for estimating the evolution of local observables under rapidly mixing Lindbladians.
This situation mimics the relationship between cluster expansion and MCMC approaches to the classical Ising models in two dimensions, where the cluster expansion breaks down before the critical temperature \cite{friedli2017statistical}, while Glauber dynamics remains rapidly mixing until the critical temperature \cite{martinelli1999glauber}.

In order to estimate an expectation value in a given state, we must execute many different shots of the measurement. It was shown in \cite{jiang2026predictingpropertiesquantumthermal, chen2026thermalexpectationestimationsingletrajectory} how to decrease this complexity by measuring expectation values from a single trajectory of the evolution, requiring an initial burn-in period of the length of the mixing time, but then only constant additional time for each new shot. Here we improve upon this by showing that the observable-specific mixing time is often constant itself, and so even this burn-in period is constant. Our work further highlights the importance of rapid mixing of open quantum systems, as it potentially allows significant qualitative improvements of the evolution time required for practical end-to-end tasks, decreasing it to a constant from as much as a polynomial time attained from only analysing the spectral gap of the generator.
\\

\paragraph*{Outlook.}
There are many potential directions for the study of observable-specific mixing times, such as showing the stability of constant mixing times under a perturbation of the Lindbladian, similar to the stability of expectation values themselves as studied in \cite{cubitt2015stability}. 
Since our results for non-interacting Lindbladians don't necessarily require the observables to be local, it is intriguing if this uniform mixing could be extended to weakly-interacting Lindbladians. 
It would also be interesting to see if one can bound the mixing times of relevant observables for Lindbladians whose full mixing times are not known, or if one could even show efficient mixing of some interesting observables for Lindbladians known to mix the full states slowly.
\\

\paragraph*{Manuscript.} The rest of the paper is divided as follows: In Section \ref{sec: observable-specific mixing times}, we define the notion of observable-specific mixing times. In Section \ref{sec: qudits}, we cover the results for qudit Lindbladians, including our main Theorem \ref{thm: constant mixing of geometrically-local observables}, classical simulability with Theorem \ref{thm: classical simulability}, and further intuition for separable Lindbladians in Subsection \ref{sec: separable qudits}. In Sections \ref{sec: fermions} and \ref{sec: bosons}, we provide results for non-interacting fermionic and bosonic systems respectively. Finally, we complement our theory with numerical simulations of the exact mixing properties in Section \ref{sec: numerics}.

\section{Observable-specific mixing times}\label{sec: observable-specific mixing times}

Open quantum system dynamics generated by a Lindbladian $\mathcal{L}^\dagger$ generalises Hamiltonian evolution by allowing Markovian dissipation of information into the environment. Unlike unitary dynamics generated by a Hamiltonian, Lindbladians generate quantum channels converging towards some steady state. Throughout this manuscript, we shall consider only those Lindbladians that generate an ergodic evolution always converging towards a unique steady state $\sigma$, which obeys $\mathcal{L}^\dagger [\sigma] = 0$ as the fixed point of the evolution. Note that $\mathcal{L}[I]=0$ provides the corresponding fixed point in the Heisenberg picture, ensuring that an evolved observable $O(t)=e^{t\mathcal{L}}[O]$ converges towards $\Tr(O\cdot \sigma)\cdot I$.
The rate of convergence of dissipative dynamics generated by such Lindbladians is then commonly studied through their mixing times:
\begin{definition}
    The mixing time of the Lindbladian $\mathcal{L}^\dagger$ is \begin{equation}
        t_\textup{mix}(\epsilon) = \inf \left\{ t \geq 0 \left|\, \forall \rho: \left\|e^{t\mathcal{L}^\dagger}[\rho] - \sigma\right\|_{\Tr} \leq \epsilon \right.\right\}\,,
    \end{equation} where $\|A\|_{\Tr}=\Tr(\sqrt{A^\dagger A}) = \sup\limits_{\|O\|\leq 1} |\Tr(O\cdot A)|$ denotes the trace norm. 
\end{definition} This definition naturally lends itself to also defining a mixing time from a specific initial state, $t_\textup{mix}^{(\rho)}(\epsilon)$, which is then clearly upper bounded by the full mixing time as $t_\textup{mix}(\epsilon) = \sup_\rho t_\textup{mix}^{(\rho)}(\epsilon)$. This general mixing time then provides a bound on the error of any measured expectation value, when initiating the dynamics in an arbitrary state, as \begin{equation}
    |\Tr\left(O \cdot (\rho(t_\textup{mix}(\epsilon)) - \sigma)\right)| \leq \|O\|\cdot \|\rho(t_\textup{mix}(\epsilon))-\sigma\|_{\Tr} \leq \epsilon \cdot \|O\|\,.
\end{equation} However, for applications in algorithmic state preparation, one is typically interested in measuring some specific, physically-relevant, expectation values in the steady state, which motivates a weaker notion of mixing.

In this work, we will be interested in mixing properties of specific observables. Corresponding to the previous bound on the difference of expectation values, we make the following definition of the mixing time for an observable $O$:
\begin{definition}\label{def:mixing time of observables}
    The mixing time of an observable $O$ under the evolution generated by $\mathcal{L}^\dagger$ and initiated in $\rho$ is \begin{align}
        t_{\textup{mix}}^{(O,\rho)}(\epsilon) &= \inf\left\{ t\geq 0 \left|\ \forall \tau \geq t: \left|\Tr\left(O \cdot (e^{\tau\mathcal{L}^\dagger}[\rho] - \sigma)\right)\right| \leq\right. \epsilon \cdot \|O\| \right\}\,,
    \end{align}
    while the mixing time of $O$ for an arbitrary initial state is
    \begin{align}
        t_{\textup{mix}}^{(O)}(\epsilon) = \sup_\rho  t_{\textup{mix}}^{(O,\rho)}(\epsilon)
        = \inf\left\{ t\geq 0 \left|\ \forall\rho, \tau \geq t: \left|\Tr\left(O \cdot (e^{\tau\mathcal{L}^\dagger}[\rho] - \sigma)\right)\right| \leq\right. \epsilon \cdot \|O\| \right\}\,.
    \end{align}
\end{definition} Note that the quantifier $\forall\tau\ge t$ in this definition is (a priori) necessary, as expectation values need not evolve monotonically for a given initial state. However, observe that \begin{align}\hspace{-1cm}
    \sup_\rho \left| \Tr\left(O \cdot (e^{t\mathcal{L}^\dagger}[\rho] - \sigma)\right)\right| &= \sup_\rho | \Tr\left(O(t) \cdot \rho - O \cdot \sigma\right)| = \max\{\lambda_\textup{max}(O(t)) - \Tr(O \cdot \sigma), \Tr(O \cdot \sigma) - \lambda_\textup{min}(O(t))\}\hspace{-1cm}\\
    &= \|O(t) - \Tr(O \cdot \sigma) \cdot I\|\,,
\end{align} and so the definition for an arbitrary initial state can be rewritten fully in the Heisenberg picture. Further, as $O(t)- \Tr(O \cdot \sigma) \cdot I = O_\perp(t)$, where $O_\perp = O - \Tr(O \cdot \sigma) \cdot I$ is the decaying part of $O$, and since the Lindbladian evolution in the Heisenberg picture is contractive with respect to the spectral norm, $\|O(t) - \Tr(O \cdot \sigma) \cdot I\|$ is actually a monotonically decreasing function of $t$, and hence we can simplify the definition as \begin{align}\label{def: simplified mixing time of O}
        t_{\textup{mix}}^{(O)}(\epsilon) &= \inf\left\{ t\geq 0 \left|\  \|e^{t\mathcal{L}}[O] - \Tr(O \cdot \sigma) \cdot I\| \leq\right. \epsilon \cdot \|O\| \right\}\,.
    \end{align}

By the motivating upper bound, we trivially have that $t_{\textup{mix}}^{(O)}(\epsilon) \leq t_{\textup{mix}}(\epsilon)$ for any observable $O$. In fact \begin{equation}t_\textup{mix}(\epsilon) = \sup_{O} t^{(O)}_\textup{mix}(\epsilon)\,,\end{equation} as we can express the trace norm in terms of expectation values, $\left\|e^{t\mathcal{L}^\dagger}[\rho] - \sigma\right\|_{\Tr} = \sup\limits_{\|O\| \leq 1} \left|\Tr\left(O \cdot (e^{t\mathcal{L}^\dagger}[\rho] - \sigma)\right)\right|$, and hence for any $t\geq \sup_{O} t^{(O)}_\textup{mix}(\epsilon)$, we get that \begin{equation}
    \left\|e^{t\mathcal{L}^\dagger}[\rho] - \sigma\right\|_{\Tr} = \sup_{\|O\| \leq 1} \left|\Tr\left(O \cdot (e^{t\mathcal{L}^\dagger}[\rho] - \sigma)\right)\right| \leq  \sup_{\|O\| \leq 1} \epsilon \cdot \|O\| = \epsilon\,,
\end{equation} and so $t_\textup{mix}(\epsilon) \leq \sup_{O} t^{(O)}_\textup{mix}(\epsilon)$.

One of the useful properties of observable-specific mixing times is how they behave for a sum of operators:

\begin{lemma}\label{lemma: mixing time of a sum}
    Given a set of observables $\{O_i\}_{i=1}^K$, the mixing time of their sum $O = \sum_i O_i$ is upper bounded like \begin{equation}
        t_\textup{mix}^{(O)}(\epsilon) \leq \sup_i t_\textup{mix}^{(O_i)}\left(\epsilon \cdot \frac{\|O\|}{\sum_j \|O_j\|}\right)\,.
    \end{equation} In particular, if each $O_i$ mixes in a constant time, and $\|O\| = \Omega (\sum_i \|O_i\|)$, then $O$ also mixes in a constant time.
\end{lemma}
\begin{proof}
    Observe that for all $t \geq \sup_i t_\textup{mix}^{(O_i)}\left(\epsilon \cdot \frac{\|O\|}{\sum_j \|O_j\|}\right)$, we have that \begin{equation}
        \|e^{\mathcal{L}t}[O] - \Tr(O\cdot \sigma) \cdot I\| \leq \sum_i \|e^{\mathcal{L}t}[O_i] - \Tr(O_i\cdot \sigma) \cdot I\| \leq \sum_i \epsilon \cdot \frac{\|O\|}{\sum_j \|O_j\|} \cdot \|O_i\| = \epsilon \cdot \|O\|\,,
    \end{equation} giving us the bound on the mixing time of $O$ as per the simplified Definition \eqref{def: simplified mixing time of O}. Further, if $\|O\| = \Omega (\sum_i \|O_i\|)$, then $\frac{\|O\|}{\sum_j \|O_j\|}$ is lower bounded by a positive system-size-independent constant. (Note that this assumption on $\|O\|$ together with triangle inequality says that $\|O\| = \Theta (\sum_i \|O_i\|)$.) Then if $O_i$ mixes in a constant time, $t_\textup{mix}^{(O_i)}(\epsilon)$ is upper bounded uniformly in the system size for any constant $\epsilon$, and so is $t_\textup{mix}^{(O_i)}\left(\epsilon \cdot \frac{\|O\|}{\sum_j \|O_j\|}\right)$, giving us the constant mixing time of $O$.
\end{proof}

 In the case that each $O_i$ mixes in a constant time, if we had explicit knowledge of their exponential contractivity $\|e^{\mathcal{L}t}[O_i] - \Tr(O_i\cdot \sigma) \cdot I\| \leq c_i e^{-\Delta_i t} \|O_i\|$, then we can similarly get the exponential contractivity for $O$ as \begin{align}
        \|e^{\mathcal{L}t}[O] - \Tr(O\cdot \sigma) \cdot I\| \leq \sum_i \|e^{\mathcal{L}t}[O_i] - \Tr(O_i\cdot \sigma) \cdot I\| \leq \sum_i c_i e^{-\Delta_i t} \|O_i\|
        \leq c_\textup{max} e^{-\Delta_\textup{min} t} \sum_i \|O_i\|\,,
    \end{align} which leads to the bound on the mixing time as \begin{equation}
        t_\textup{mix}^{(O)}(\epsilon) \leq \frac{1}{\Delta_\textup{min}} \log \left( \frac{c_\textup{max}}{\epsilon} \cdot \frac{\sum_i \|O_i\|}{\|O\|} \right)\,.
    \end{equation}

\newpage
\section{Qudit spin systems}\label{sec: qudits}

\begin{theorem}\label{thm: constant mixing of geometrically-local observables}
    Consider a uniform family of quasi-local and rapidly mixing Lindbladians $\mathcal{L}$, with a uniform lower bound on their spectral gaps, acting on a qudit Hilbert space. Any observable $O$ which can be written as a sum of geometrically-local terms will mix in a constant (system-size-independent) time when evolved under $\mathcal{L}$, \begin{equation}
        t_\textup{mix}^{(O)}(\epsilon) =\mathcal{O}(\log(1/\epsilon))\,.
    \end{equation}
\end{theorem}
\begin{proof}
    Firstly note that a uniform family of Lindbladians is characterised by the Lindbladian acting on the infinite lattice $\Lambda$ together with its restrictions to finite sizes $S \subseteq \Lambda$ with appropriate boundary conditions \cite[Definition 3]{cubitt2015stability}. In the case of quantum Gibbs samplers, this simply corresponds to restricting the Hamiltonian of the system and the set of jump operators to the set of sites $S$. Here we characterise rapid mixing by its contractivity (in the Heisenberg picture) as follows: \begin{equation}
        \left\|e^{t\mathcal{L}_S}[O_S] - \Tr_S\left(e^{t\mathcal{L}_S}[O_S]\right) \frac{I}{d^{|S|}} \right\| \leq \operatorname{poly}(|S|)\cdot e^{-\Delta t}\cdot\|O_S\|\,,
    \end{equation} where $d$ is the local qudit dimension. In particular, this also means that we require the spectral gap of $\mathcal{L}$ to be uniformly lower bounded by $\Delta$ (in contrast to e.g.~closing polylogarithmically with system size). Further, quasi-locality here means that $\mathcal{L}$ has exponentially decaying interactions, which is also a common property for quantum Gibbs samplers with local jump operators \cite{chen2025efficient, ding2025efficient}. 

    Now consider a geometrically-local term $O_A$ supported on a finite local region $A$, and also the ball $B_A(r)$ of radius $r$ around $A$. We can split the contractivity of $O_A$ under $\mathcal{L}$ into a contribution coming from $B_A(r)$ and from $\Lambda\backslash B_A(r)$ as follows: \begin{align}
        \left\| e^{t\mathcal{L}}[O_A] - \Tr(e^{t\mathcal{L}}[O_A]) \frac{I}{d^n} \right\| &\leq  \left\| e^{t\mathcal{L}}[O_A] - e^{t\mathcal{L}_{B_A(r)}}[O_A] - \Tr(e^{t\mathcal{L}}[O_A]-e^{t\mathcal{L}_{B_A(r)}}[O_A]) \frac{I}{d^n} \right\|\\ &\qquad + \left\| e^{t\mathcal{L}_{B_A(r)}}[O_A] - \Tr(e^{t\mathcal{L}_{B_A(r)}}[O_A]) \frac{I}{d^n} \right\| \\
        &\leq 2  \left\| e^{t\mathcal{L}}[O_A] - e^{t\mathcal{L}_{B_A(r)}}[O_A]\right\| + \left\| e^{t\mathcal{L}_{B_A(r)}}[O_A] - \Tr_{B_A(r)}(e^{t\mathcal{L}_{B_A(r)}}[O_A]) \frac{I}{d^{|{B_A(r)|}}} \right\|\,,
    \end{align} where we've also used that $
        \frac{1}{d^n} \left| \Tr(e^{t\mathcal{L}}[O_A]-e^{t\mathcal{L}_{B_A(r)}}[O_A])\right| \leq   \left\| e^{t\mathcal{L}}[O_A] - e^{t\mathcal{L}_{B_A(r)}}[O_A]\right\|$. 
    Since the Lindbladian $\mathcal{L}$ is quasi-local, its dynamics obeys Lieb-Robinson bounds \cite{Barthel2012,Nachtergaele2011,Poulin2010}, and it can be localised using \cite[Lemma 11]{cubitt2015stability} as \begin{equation}
        \left\| e^{t\mathcal{L}}[O_A] - e^{t\mathcal{L}_{B_A(r)}}[O_A]\right\| \leq |A| \cdot J \cdot \frac{e^{vt}-1-vt}{v} \cdot e^{-\gamma r} \cdot \|O_A\|\,,
    \end{equation} where $J$, $v$, and $\gamma$ are system-size-independent constants. Hence in order to make the first contribution sufficiently small, \begin{equation}
        2  \left\| e^{t\mathcal{L}}[O_A] - e^{t\mathcal{L}_{B_A(r)}}[O_A]\right\| \leq \frac{\epsilon}{4} \|O_A\|\,,
    \end{equation} we just need to choose a sufficiently large radius. Setting \begin{equation}
        |A| \cdot J \cdot \frac{e^{vt}-1-vt}{v} \cdot e^{-\gamma r} \overset{\text{set}}{=} \frac{\epsilon}{8}
    \end{equation} gives \begin{equation}
        r = \frac{1}{\gamma} \log\left( 8|A| J\frac{e^{vt}-1-vt}{v}\cdot \frac{1}{\epsilon} \right) = \mathcal{O}(t+\log(1/\epsilon))\,.
    \end{equation} Regarding the second contribution, by the assumption on rapid mixing of $\mathcal{L}$, we have that \begin{equation}
        \left\| e^{t\mathcal{L}_{B_A(r)}}[O_A] - \Tr_{B_A(r)}(e^{t\mathcal{L}_{B_A(r)}}[O_A]) \frac{I}{d^{|{B_A(r)|}}} \right\| \leq \operatorname{poly}(|B_A(r)|) \cdot e^{-\Delta t} \cdot \|O_A\|\,.
    \end{equation} Since we are in a local setting, say on a lattice in $D$ dimensions, we have that $|B_A(r)| = \mathcal{O}(r^D) = \operatorname{poly}(r)$. Hence choosing $r = \mathcal{O}(t)$ implies $|B_A(r)| = \operatorname{poly}(t)$, and hence \begin{equation}
        \left\| e^{t\mathcal{L}_{B_A(r)}}[O_A] - \Tr_{B_A(r)}(e^{t\mathcal{L}_{B_A(r)}}[O_A]) \frac{I}{d^{|{B_A(r)|}}} \right\| \leq \operatorname{poly}(t,\log(1/\epsilon)) \cdot e^{-\Delta t}\cdot \|O_A\|\,.
    \end{equation} Since the exponential decay in $t$ asymptotically beats the polynomial growth, the right hand side will eventually start decreasing towards $0$, and hence there will exist some $t_\star = \mathcal{O}(\log(1/\epsilon))$, independent of the system size, such that for any $t\geq t_\star$ we have\begin{equation}
        \left\| e^{t\mathcal{L}_{B_A(r)}}[O_A] - \Tr_{B_A(r)}(e^{t\mathcal{L}_{B_A(r)}}[O_A]) \frac{I}{d^{|{B_A(r)|}}} \right\| \leq \frac{\epsilon}{4}\cdot \|O_A\|\,.
    \end{equation} Together we have that for all $t\geq t_\star$, \begin{equation}
        \sup_\rho \left|\Tr\left(O_A \cdot (e^{t\mathcal{L}^\dagger}[\rho] - \sigma)\right)\right| \leq 2 \left\| e^{t\mathcal{L}}[O_A] - \Tr(e^{t\mathcal{L}}[O_A]) \frac{I}{d^n} \right\| \leq \epsilon \cdot \|O_A\|\,,
    \end{equation} where we've also used that \begin{align}
        \left|\Tr\left(O_A \cdot (e^{t\mathcal{L}^\dagger}[\rho] - \sigma)\right)\right| = \left|\Tr\left((e^{t\mathcal{L}}[O_A] + c I)\cdot (\rho - \sigma)\right)\right|
        \leq \|\rho-\sigma\|_{\Tr}\cdot\|e^{t\mathcal{L}}[O_A] + c I\|
        \leq 2\|e^{t\mathcal{L}}[O_A] + c I\|
    \end{align} for any scalar $c$, and so the term $O_A$ mixes in a constant time \begin{equation}
        t_\textup{mix}^{(O_A)}(\epsilon) \leq t_\star = \mathcal{O}(\log(1/\epsilon))\,.
    \end{equation}

    Finally, we wish to lift the constant mixing time of local terms $O_A$ to the observable $O = \sum_A O_A$. To do so, we can utilise Lemma \ref{lemma: mixing time of a sum}, which just requires upper bounding $\frac{\sum_A \|O_A\|}{\|O\|}$ by a constant. Note that, for geometrically-local observables on a Hilbert space of qubits, such a bound has been previously shown for Pauli expansions. Specifically, for $O = \sum_i O_i$, where $O_i$ are multiples of distinct Paulis, so that $\sum_i \|O_i\|$ is the Pauli-$1$ norm, by \cite[Theorem 2]{lewis2024improved}, we have that \begin{equation}
        \frac{\sum_i \| O_i\|}{\|O\|} \leq 2^D \cdot V \cdot 4^V\,,
    \end{equation} where $D$ is the dimension of the lattice and $V$ is the maximal volume of local interactions of the observable, defined as $V = \prod_{k=1}^D R_k$, where $R_k$ is the maximal range of interactions in the direction $k$. Observe that we can map a geometrically-local observable on qudits onto a geometrically-local observable of qubits, where the maximal range of interactions will get appropriately rescaled by a factor of $\lceil \log_2(d)\rceil$. Without loss of generality, we can then consider the expansion of $O$ in Paulis and apply the previous result in order to get \begin{equation}
        \frac{\sum_A \|O_A\|}{\|O\|} \leq c\,,
    \end{equation} where the constant $c$ depends only on the dimension of the lattice $D$, the local qudit dimension $d$, and the maximal range of interactions $R$. Combining this bound with Lemma \ref{lemma: mixing time of a sum} and the constant mixing time of the individual terms $O_A$, we find that the observable $O$ mixes in a constant time, yielding the result of this theorem.
\end{proof}

\begin{corollary}[Gibbs states of weakly-interacting systems]\label{cor: constant mixing weakly-interacting Gibbs}
    For a quasi-local qudit system $H = H_0+\lambda V$, where $H_0$ is a separable (1-local) Hamiltonian, at any fixed temperature, there exists a maximal interaction strength $\lambda_\textup{max}$, independent of the system size, below which quantum Gibbs samplers thermalise geometrically-local observables in a constant time.
\end{corollary}
    It was proven in \cite[Theorem III.3]{smid2025rapid} that such quantum Gibbs samplers with local jump operators are quasi-local and rapidly mixing up to $\lambda_\textup{max}$ given by \cite[Corollary III.3.1]{smid2025rapid}.

\begin{corollary}[Gibbs states at high temperatures]\label{cor: constant mixing  high-temperature Gibbs}
    For a geometrically-local qudit Hamiltonian $H$, there exists a critical inverse temperature $\beta_\star$, independent of the system size, below which quantum Gibbs samplers thermalise geometrically-local observables in a constant time.
\end{corollary}
    It was proven in \cite[Theorem 1]{Rouze2026} that for a $(k,l)$-local Hamiltonian on a $D$-dimensional lattice, $H = \sum_X h_X$, with $\|h_X\|\leq h$, if we define $J =hkl$, then for any $\beta < \beta_\star = \frac{1}{615^D \cdot 2J}$ such quantum Gibbs samplers with local jump operators are quasi-local and rapidly mixing. Here we again use the fact that short-range qudit Hamiltonians can be mapped to short-range qubit Hamiltonians. Note that this result is straightforwardly generalisable to quasi-local Hamiltonians.

\begin{corollary}[Gibbs states of 1D systems]\label{cor: constant mixing 1D systems}
    For a geometrically-local qudit Hamiltonian $H$ on a 1D chain, at any fixed temperature, the quantum Gibbs sampler of \cite{kim2026rapid} thermalises geometrically-local observables in a constant time.
\end{corollary}
    It was proven in \cite{kim2026rapid} that this Lindbladian is quasi-local and rapidly mixing for any spin chains with finite-range interactions. Note that we can again map short-range qudit systems to short-range qubit systems.

\begin{corollary}[Ground states of weakly-interacting systems]\label{cor: constant mixing ground states}
    For a geometrically-local qubit Hamiltonian $H = -\sum_i Z_i + \lambda V$, there exists a maximal interaction strength $\lambda_\textup{max}$, independent of the system size, below which Lindbladians for quantum ground state preparation thermalise geometrically-local observables in a constant time.
\end{corollary}
    It was proven in \cite[Theorem 8]{Zhan2026} that such Lindbladians are also rapidly mixing and quasi-local. Since our result doesn't require the steady state to be full-rank, it is also applicable in this setting. Note that this result is likely extendable to general quasi-local weakly-interacting systems as in Corollary \ref{cor: constant mixing weakly-interacting Gibbs} using similar techniques to \cite{smid2025rapid}.

\begin{remark}
    We note that the proof of Theorem \ref{thm: constant mixing of geometrically-local observables} requires rapid mixing. If we only had fast mixing from a constant spectral gap, the contribution from the inner region would scale with $e^{\mathcal{O}(t^D)-\Delta t}$ instead of $\operatorname{poly}(t)\cdot e^{-\Delta t}$. For any dimension $D\geq 2$, this would be an increasing function of $t$ and wouldn't decay. 
\end{remark}

\enlargethispage{1cm}
\subsection{Classical simulability}

We can observe that results of Theorem \ref{thm: constant mixing of geometrically-local observables} are immediately applicable to lowering the complexity of simulating such Lindbladians and estimating geometrically-local expectation values: 

\begin{theorem}\label{thm: classical simulability}
    For a quasi-local and rapidly mixing Lindbladian $\mathcal{L}$ with a unique steady-state $\sigma$, there exists a classical algorithm (given by Algorithm \ref{alg: Estimating expectation values}) for evaluating the expectation values $\Tr(O\cdot \sigma)$ for any observable $O$, which is expressible as a sum of geometrically-local terms, with time complexity scaling like $\mathcal{O} \left(n\cdot e^{\mathcal{O}\left(\log(1/\epsilon)^D\right)}\right)$, where $n$ is the system size, $\epsilon$ the desired relative error on the expectation value, and $D$ the dimension of the system. Moreover, this scaling is optimal in the system size $n$, over which the algorithm is fully parallelisable.
\end{theorem}

\renewcommand{\algorithmicrequire}{\textbf{Input:}}
\renewcommand{\algorithmicensure}{\textbf{Output:}}
\begin{algorithm}[H]
  \caption{Evaluating geometrically-local expectation values}
  \label{alg: Estimating expectation values}
   \begin{algorithmic}[1]
        \Require Description of $\mathcal{L}$, local expansion of $O=\sum_{A \in S} O_A$, accuracy $\epsilon$
        \Ensure Estimate for $\Tr(O\cdot \sigma)$ with an error of at most $\epsilon\cdot \|O\|$

        \State Calculate $t^*$ as the maximal solution to Equation \eqref{eqn: solve for t^*}
        \State Calculate $r^*$ from Equation \eqref{eqn: solution for r^*}
        \State $V = 0$
        \For{$A \in S$}
            \State Simulate $e^{t^* \cdot\mathcal{L}_{B_A(r^*)}}[O_A]$ classically \Comment{Can use na\"ive matrix exponentiation, or e.g.~\cite{Smid_GitHub_GibbsSampling}}
            \State $V \text{ += } \frac{1}{d^{|B_A(r^*)|}} \cdot \Tr\left( e^{t^*\cdot \mathcal{L}_{B_A(r^*)}}[O_A] \right)$
        \EndFor
        \State \Return $V$\Comment{The estimate for $\Tr(O\cdot \sigma)$ is $V$}

   \end{algorithmic}
\end{algorithm}

\begin{proof}
    Firstly, we express $O$ in an orthogonal basis $O=\sum_{A \in S} O_A$ with local terms $O_A$ (say, multiples of Paulis).
    Note that for a geometrically-local $O$, we have that $|S| = \mathcal{O}(n)$. 
    From the description of $\mathcal{L}$, we know the parameters $J,v$, and $\gamma$ appearing in the Lieb-Robinson bound, \begin{equation}
        \left\| e^{t\mathcal{L}}[O_X] - e^{t\mathcal{L}_{B_X(r)}}[O_X]\right\| \leq |X| \cdot J \cdot \frac{e^{vt}-1-vt}{v} \cdot e^{-\gamma r} \cdot \|O_X\|\,,
    \end{equation} and we also know the polynomial $q(x)$ and a lower bound $\Delta$ on the spectral gap appearing in the contractivity of $\mathcal{L}$, \begin{equation}
        \left\|e^{t\mathcal{L}_X}[O_X] - \Tr_X\left(e^{t\mathcal{L}_X}[O_X]\right) \frac{I}{d^{|X|}} \right\| \leq q(|X|)\cdot e^{-\Delta t}\cdot\|O_X\|\,.
    \end{equation}
    Let $B_A(r)$ be the ball of radius $r$ around the set $A$, and note that on a $D$-dimensional lattice, we have $|B_A(r)| = \mathcal{O}(r^D)$. Define \begin{equation}
        r(t) = \frac{1}{\gamma} \log\left( \left(\max_{A \in S}|A|\right) \cdot  J\frac{e^{vt}-1-vt}{v}\cdot \frac{9}{\epsilon} \cdot \frac{\sum_{A} \|O_A\|}{\|O\|} \right)\,,
    \end{equation} which guarantees \begin{equation}
        \left\| e^{t\mathcal{L}}[O_A] - e^{t\mathcal{L}_{B_A(r(t))}}[O_A]\right\| \leq \frac{\epsilon}{9} \cdot \frac{\|O\|}{\sum_B \|O_B\|} \cdot \|O_A\|\,.
    \end{equation} 
    Note that $\frac{\sum_{A} \|O_A\|}{\|O\|}$ is upper bounded by an $n$-independent value as per Theorem \ref{thm: constant mixing of geometrically-local observables}.
    Then let $t^*$ be the maximal solution to the equation \begin{equation}\label{eqn: solve for t^*}
        q\left(\max_{A \in S}|B_A(r(t))|\right) \cdot e^{-\Delta t} = \frac{2}{9} \epsilon\,.
    \end{equation} Further denote \begin{equation}\label{eqn: solution for r^*}
        r^* = r(t^*)\,.
    \end{equation}
    We then wish to approximate $\Tr(O\cdot \sigma) \approx \sum_{A\in S} \frac{1}{d^{|B_A(r^*)|}} \cdot \Tr_{B_A(r^*)}\left( e^{t^* \cdot\mathcal{L}_{B_A(r^*)}}[O_A] \right)$. This gives the following error: \begin{align}\hspace{-2cm}
        \left| \Tr(O\cdot \sigma) - \sum_{A\in S} \frac{1}{d^{|B_A(r^*)|}} \cdot \Tr_{B_A(r^*)}\left( e^{t^* \cdot\mathcal{L}_{B_A(r^*)}}[O_A] \right)\right| 
        &\leq  \left| \Tr(O\cdot \sigma) - \frac{1}{d^n}\Tr\left( e^{t^*\cdot \mathcal{L}}[O] \right) \right|\\ 
        &\quad + \sum_{A\in S} \left| \frac{1}{d^{|B_A(r^*)|}} \cdot \Tr_{B_A(r^*)}\left( e^{t^* \cdot\mathcal{L}_{B_A(r^*)}}[O_A] \right) - \frac{1}{d^n}\Tr\left( e^{t^*\cdot \mathcal{L}}[O_A] \right) \right|\,.\hspace{-2cm}
    \end{align} Now note that we have \begin{align}
        \left| \Tr(O\cdot \sigma) - \frac{1}{d^n}\Tr\left( e^{t^*\cdot \mathcal{L}}[O] \right) \right| = \left| \frac{1}{d^n}\Tr\left( e^{t^*\cdot \mathcal{L}}[O] - \Tr(O\cdot \sigma)\cdot I \right)\right| \leq \left\| e^{t^*\cdot \mathcal{L}}[O] - \Tr(O\cdot \sigma)\cdot I \right\| \leq \frac{8}{9} \cdot \epsilon \cdot \|O\|\,,
    \end{align} where the last inequality follows from Theorem \ref{thm: constant mixing of geometrically-local observables} as $t_\textup{mix}^{(O)}\left(\frac{8}{9}\epsilon \right) \leq t^*$. Similarly, we have that \begin{align}\hspace{-1cm}
        \left| \frac{1}{d^{|B_A(r^*)|}} \cdot \Tr_{B_A(r^*)}\left( e^{t^* \cdot\mathcal{L}_{B_A(r^*)}}[O_A] \right) - \frac{1}{d^n}\Tr\left( e^{t^*\cdot \mathcal{L}}[O_A] \right) \right| \leq \left\| e^{t^*\mathcal{L}}[O_A] - e^{t^*\mathcal{L}_{B_A(r^*)}}[O_A]\right\| \leq \frac{\epsilon}{9} \cdot \frac{\|O\|}{\sum_B \|O_B\|} \cdot \|O_A\|\,.\hspace{-1cm}
    \end{align}
    Altogether, we get that \begin{align}
        \left| \Tr(O\cdot \sigma) - \sum_{A\in S} \frac{1}{d^{|B_A(r^*)|}} \cdot \Tr_{B_A(r^*)}\left( e^{t^* \cdot\mathcal{L}_{B_A(r^*)}}[O_A] \right)\right|  \leq \frac{8}{9} \cdot \epsilon \cdot \|O\| +\sum_{A \in S} \frac{\epsilon}{9} \cdot \frac{\|O\|}{\sum_B \|O_B\|} \cdot \|O_A\| = \epsilon \cdot \|O\|\,,
    \end{align} and so the estimate indeed gives the desired error.

    Now observe that the time complexity is simply $|S|$ times the complexity of simulating $e^{t^* \cdot\mathcal{L}_{B_A(r^*)}}[O_A]$. As per Theorem \ref{thm: constant mixing of geometrically-local observables}, we have that $t^*$ always exists and scales like $t^* = \mathcal{O}(\log(1/\epsilon))$ (independently of $n$), and then $r^* = \mathcal{O}(\log(1/\epsilon))$ with $N\coloneq |B_A(r^*)| = \mathcal{O}(\log(1/\epsilon)^D)$. Note that we can classically calculate the exponential of a qudit Lindbladian acting on $N$ sites in time $\mathcal{O}(d^{6N})$ in the worst case (which is most likely very suboptimal). Hence the evaluation of $\frac{1}{d^{|B_A(r^*)|}} \cdot \Tr\left( e^{t^*\cdot \mathcal{L}_{B_A(r^*)}}[O_A] \right)$ takes at most $\mathcal{O}\left(d^{6\cdot \mathcal{O}(\log(1/\epsilon)^D)}\right) = \mathcal{O}\left(e^{\mathcal{O}(\log(1/\epsilon)^D)}\right)$ time, making the overall complexity of the procedure $\mathcal{O}\left(|S| \cdot e^{\mathcal{O}(\log(1/\epsilon)^D)}\right) = \mathcal{O}\left(n \cdot e^{\mathcal{O}(\log(1/\epsilon)^D)}\right)$.
\end{proof}

\begin{remark}
    Due to Corollary \ref{cor: constant mixing 1D systems}, this classical algorithm is fully efficient for all short-range Hamiltonians on 1D chains, with its complexity scaling like $\mathcal{O}(n\cdot\operatorname{poly}(1/\epsilon))$.
\end{remark}

\begin{remark}
    The same strategy is directly applicable to the quantum algorithms for simulating such Lindbladians, like quantum Gibbs samplers, lowering their end-to-end complexity for estimating steady-state expectation values of geometrically-local observables to $\mathcal{O}(n\cdot\operatorname{poly}(1/\epsilon))$. The quantum simulation then still provides a superpolynomial speed-up with respect to the accuracy $\epsilon$ for any $D\geq 2$.
\end{remark}

In order to demonstrate Algorithm \ref{alg: Estimating expectation values} in practice, we use it for calculating the Gibbs state energy of a 1D transverse-field Ising model (TFIM) given by Equation \eqref{eqn: Hamiltonian TFIM} via a corresponding quantum Gibbs sampler. On Figure \ref{fig:classical algorithm error}, we choose a truncation radius $r=5$ and an evolution time $t=3$, and evaluate the relative error $\epsilon$ of this calculation compared to the exact Gibbs state energy for large chains with up to $n=10000$ sites. Note that the 1D TFIM maps under the Jordan-Wigner transformation to a free fermionic system, and as such we can efficiently evaluate its Gibbs state energy to make this comparison. Further, we plot the error obtained from simply taking the truncated steady-state values to radius $r=5$ (corresponding to evolution time $t=\infty$). In either case, we see an error which is upper bounded uniformly in the system size, agreeing with Theorem \ref{thm: classical simulability}, which is a manifestation of the absence of long-range correlations in the steady states of rapidly mixing Lindbladians \cite{Kastoryano2013decay}.

\begin{figure}[H]
    \centering
    \hspace{-1cm}\includegraphics[width=1.05\linewidth]{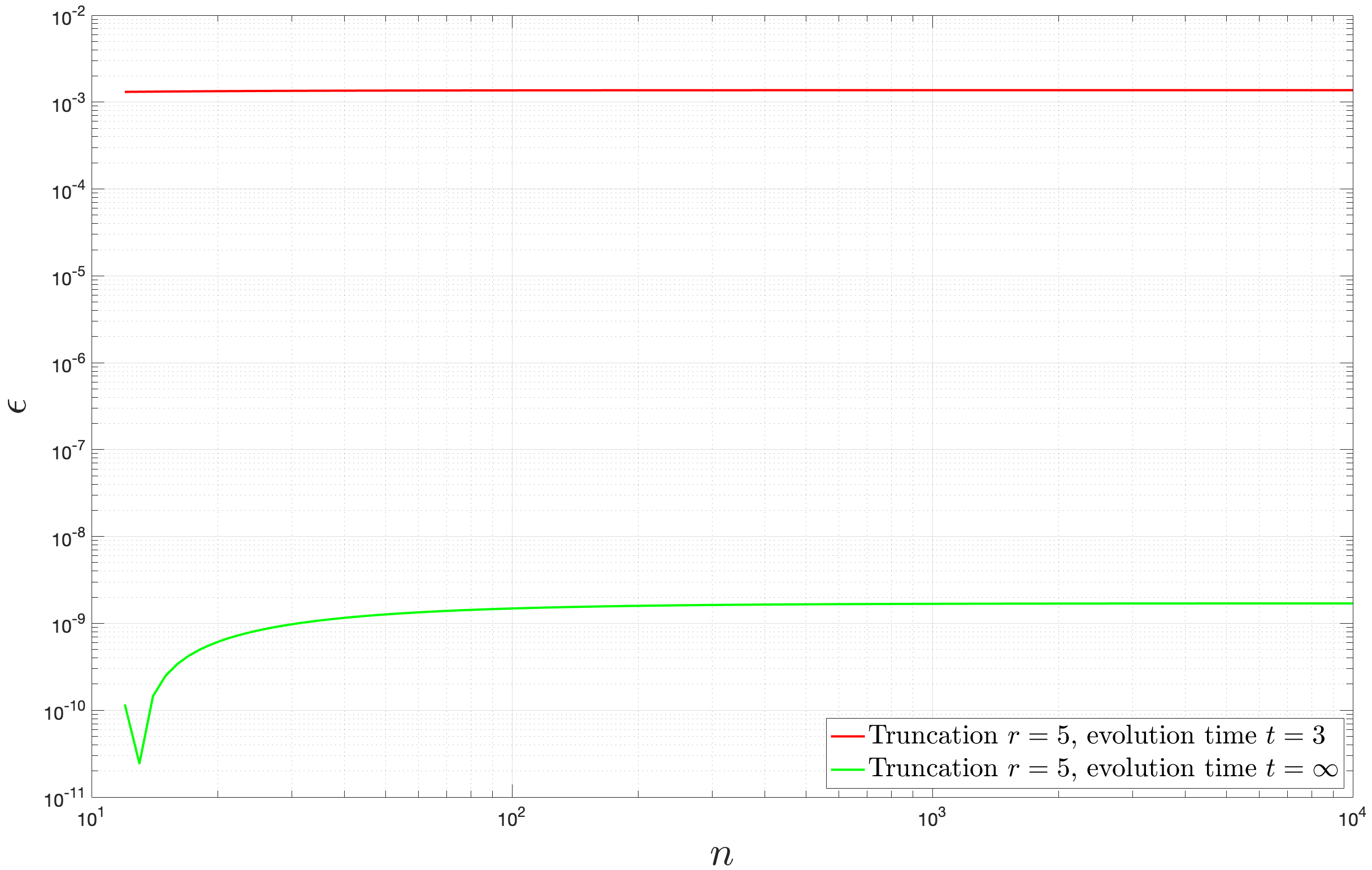}
    \caption{Relative error $\epsilon$ of Algorithm \ref{alg: Estimating expectation values} on the Gibbs state energy of the 1D transverse-field Ising model (Equation \eqref{eqn: Hamiltonian TFIM}) with $J=0.5$, $h=1$, at $\beta =1$. 
    Here we compare the error obtained by taking truncation radius $r=5$ and evolution time $t=3$ to the one stemming from the steady-state values with truncation $r=5$ for chains with up to $n=10000$ sites.}
    \label{fig:classical algorithm error}
\end{figure}

\subsection{Separable Lindbladians}\label{sec: separable qudits}

To obtain further insights into which observables mix in a constant time and which don't, we consider the simple case of $1$-local (separable) Lindbladians acting on qudits with local dimension $d$, providing explicit large families of observables which mix in a constant time, as well as an example of an observable mixing in logarithmic time (Example \ref{example: observable mixing in log time}).

\begin{prop}\label{prop: single site observables under separable Lindbladian}
    Consider a separable Lindbladian $\mathcal{L} = \sum_{i=1}^n \mathcal{L}_i$.
    Single site observables $O_i$ and $1$-local observables $O = \sum_i O_i$ mix in a constant time. Specifically, when $\mathcal{L}$ obeys KMS detailed-balance with $\sigma = \bigotimes_i \sigma_i$, the mixing time is bounded like \begin{equation}
        t_\textup{mix}^{(O)}(\epsilon) \leq \frac{1}{\Delta_0} \log\left(\max_i \left\| \sigma_i^{-1/2}\right\|\cdot\frac{2}{\epsilon} \right)\,,
    \end{equation} where $\Delta_0 = \min_i \Delta_i$ is the spectral gap of $\mathcal{L}$.
\end{prop}
\begin{proof}
    Note that for a single-site $O_i$, we have $e^{\mathcal{L} t}[O_i] = e^{\mathcal{L}_i t}[O_i]$. Using \cite[Theorem 3.3]{szehr2015spectral}, we have the (rather pessimistic) upper bound \begin{equation}
        \left\|e^{t\mathcal{L}}[O_i]-\Tr(O_i\cdot \sigma) \cdot I\right\| = \left\|e^{t\mathcal{L}_i}[O_i]-\Tr(O_i\cdot \sigma_i) \cdot I\right\| \leq 2e^{-\Delta_i t} \cdot t^{d^2-1} \cdot e^{\Delta_i (d^2-1)} (e^{-\Delta_i}+2)^{d^2-1}\cdot\|O_i\|\,,
    \end{equation} where $\Delta_i$ is the spectral gap of $\mathcal{L}_i$, meaning that all non-zero eigenvalues $\lambda^{(i)}$ of $\mathcal{L}_i$ obey $\mathrm{Re}(\lambda^{(i)}) \leq -\Delta_i$. Note that we can upper bound $e^{-\Delta_i t} \cdot t^{d^2-1}$ by a pure decaying exponential like \begin{equation}
    e^{-\Delta_i t} \cdot t^{d^2-1} \leq e^{-(\Delta_i-\delta) t} \cdot \max_{t\geq 0} \left(\frac{t^{d^2-1}}{e^{\delta \cdot t}}\right)
    \end{equation} for any $\delta \in (0,\Delta_i)$. This gives us the mixing time bound \begin{equation}
         t_\textup{mix}^{(O_i)}(\epsilon) \leq \min_{\delta \in (0,\Delta_i)} \left(\frac{1}{\Delta_i-\delta}\cdot \log\left(\frac{2}{\epsilon} \cdot \max_{t\geq 0} \left(\frac{t^{d^2-1}}{e^{\delta \cdot t}}\right) \cdot e^{\Delta_i (d^2-1)} (e^{-\Delta_i}+2)^{d^2-1}\right)\right)\,.
    \end{equation}
    Then we have $\|O\| = \sum_i \|O_i\|$ due to the $1$-locality, and hence by Lemma \ref{lemma: mixing time of a sum} we have that \begin{equation}
         t_\textup{mix}^{(O)}(\epsilon)  \leq \sup_i  t_\textup{mix}^{(O_i)}(\epsilon)\,.
    \end{equation}
    
    The situation is much clearer when $\mathcal{L}$ obeys a detailed-balance condition with $\sigma$, for which we have that \begin{align}
        \left\| e^{\mathcal{L}_i t}[O_i] - \Tr_i(O_i \sigma_{i})\cdot I \right\| \leq 2e^{-\Delta_i t} \left\|\sigma_{i}^{-1/2}\right\| \cdot \|O_i\|\,.
    \end{align} This leads to the mixing time bound \begin{equation}
        t_\textup{mix}^{(O)}(\epsilon) \leq \frac{1}{\Delta_0} \log\left(\max_i \left\| \sigma_i^{-1/2}\right\|\cdot\frac{2}{\epsilon} \right)\,.
    \end{equation} In particular, for quantum Gibbs samplers corresponding to the $1$-local qudit Hamiltonian $H=\sum_i h_i$ at inverse temperature $\beta$, using \cite[Lemma III.1]{smid2025rapid}, we get that
    \begin{equation}
        t_\textup{mix}^{(O)}(\epsilon) \leq \frac{1}{\Delta_0} \log\left(\frac{2d^{1/2}e^{\beta\cdot \Delta E_\textup{max}/2}}{\epsilon} \right)\,,
    \end{equation} where $\Delta E_\textup{max} = \max_i \Delta E_i = \max_i (\lambda_\textup{max}(h_i) - \lambda_\textup{min}(h_i))$ is the maximal spectral range of $h_i$'s.
\end{proof}

\begin{prop}
    Consider a separable Lindbladian $\mathcal{L} = \sum_{i=1}^n \mathcal{L}_i$ with a full-rank steady state $\sigma$. Any product observable $O = \bigotimes\limits_{i=1}^n O_i$ mixes in a constant time when evolved under $\mathcal{L}$.
\end{prop}
\begin{proof}
    Without loss of generality, assume $\|O_i\| = 1$ for all $i$, and hence $\|O\|=1$. Further, denote the number of sites $i$ such that $O_i \neq \pm I$ (where we consider only hermitian $O_i$) as $k(n)$, where $k(n) \in [n] = \{1,\dots,n\}$ might or might not depend on $n$. We can restrict our view to these $k(n)$ sites, as the others won't be affected by $\mathcal{L}$ and won't contribute to the mixing. We have that $O(t) = e^{t\mathcal{L}}[O] = \bigotimes\limits_{i=1}^{k(n)} O_i(t)$ due to the separability of $\mathcal{L}$. Denote by $o_i = \Tr(O_i \cdot\sigma)$ the expectation value of $O_i$ in the steady state, i.e.~have that $O_i(t) \to o_i \cdot I$ as $t \to \infty$. As $\|O_i\| = 1$, we have that $|o_i| \leq 1$, and since $O_i \not\propto I$ and $\sigma$ is full-rank, we actually get a strict inequality $|o_i| < 1$. Then we have that $O(t) \to \prod_i o_i \cdot I$, and so the mixing time of $O$ is the first time $t$ such that $\|O(t) - \prod_i o_i \cdot I\| \leq \epsilon$. Each $O_i$ also mixes in a constant time as per Proposition \ref{prop: single site observables under separable Lindbladian}, $\|O_i(t)-o_i \cdot I\| \leq c_i e^{-\Delta_i t} \|O_i\| = c_i e^{-\Delta_i t}$, and so we have that \begin{equation}
        \operatorname{spec}(O_i(t)) \subseteq \left[o_i - c_i e^{-\Delta_i t}, o_i + c_i e^{-\Delta_i t}\right]\,.
    \end{equation}
        
    As the first case, assume that $O$ has zero expectation value in $\sigma$, i.e.~that there exists $i$ with $o_i=0$, say $o_j=0$. Then this spectral norm is simply \begin{equation}
        \|O(t) - \prod_i o_i \cdot I\| = \|O(t)\| = \prod_i \|O_i(t)\|\,.
    \end{equation} Due to contractivity of the evolution, we have that $\|O_i(t)\| \leq \|O_i\|=1$, and so \begin{equation}
        \prod_i \|O_i(t)\| = \|O_j(t)\| \cdot  \prod_{i\neq j} \|O_i(t)\| \leq \|O_j(t)\| \leq c_j e^{-\Delta_j t} \overset{\text{set}}{\leq} \epsilon\,,
    \end{equation} and so we get that $t_\textup{mix}^{(O)}(\epsilon) \leq \frac{1}{\Delta_j} \log\left(\frac{c_j}{\epsilon}\right)$. While this bound is sufficient to obtain a constant mixing time, for future convenience, we would also like to understand the scaling with respect to $k(n)$, which this bound doesn't cover. For this, denote the number of terms with $o_i = 0$ by $l$, so that the remaining $k(n)-l$ non-trivial terms have $o_i > 0$ (positivity can be assumed without loss of generality as explained below in the second case, where all $o_i \neq 0$). Hence we obtain the bound \begin{equation}
         \prod_i \|O_i(t)\| \leq (c_\textup{max}\cdot e^{-\Delta_\textup{min} t})^l \cdot (o_\textup{max}+c_\textup{max}\cdot e^{-\Delta_\textup{min} t})^{k(n)-l}\,.
    \end{equation} Then for all $t\geq \frac{1}{\Delta_\textup{min}}\log\left(\frac{2c_\textup{max}}{1-o_\textup{max}}\right)$, we have that \begin{equation}
        (c_\textup{max}\cdot e^{-\Delta_\textup{min} t})^l \cdot (o_\textup{max}+c_\textup{max}\cdot e^{-\Delta_\textup{min} t})^{k(n)-l} \leq  (c_\textup{max}\cdot e^{-\Delta_\textup{min} t})^l \cdot \left(\frac{1+o_\textup{max}}{2}\right)^{k(n)-l} \overset{\text{set}}{\leq} \epsilon\,.
    \end{equation} Setting this smaller to $\epsilon$ gives \begin{equation}
        t \geq \frac{1}{\Delta_\textup{min}}\log\left(c_\textup{max} \left(\frac{1+o_\textup{max}}{2}\right)^{\frac{k(n)}{l}-1} \cdot \frac{1}{\epsilon^{1/l}}\right)\,,
    \end{equation} which means we obtain the following bound on the mixing time: \begin{equation}
        t_\textup{mix}^{(O)}(\epsilon) \leq \frac{1}{\Delta_\textup{min}}\log\left(c_\textup{max}  \cdot \max\left\{\frac{2}{1-o_\textup{max}},\,\left(\frac{1+o_\textup{max}}{2}\right)^{\frac{k(n)}{l}-1} \cdot \frac{1}{\epsilon^{1/l}} \right\}\right)\,.
    \end{equation} Now observe, that for a fixed $\epsilon$, as $k(n) \to \infty$ (independently of the behaviour of $l \in [k(n)]$), we have that this upper bound behaves like $\sim \frac{1}{\Delta_\textup{min}}\log\left(\frac{2c_\textup{max}}{1-o_\textup{max}}\right)$, which is interestingly independent of $\epsilon$. This behaviour will be expanded upon at the end of this proof. 

    As the second case, we can assume that $o_i > 0$ for all $i$, as if there were negative values of some $o_i$'s, we can simply map the corresponding observables to their negatives, $O_i \mapsto -O_i$, which doesn't change the mixing properties. The spectral norm of interest is then simply \begin{align}
        \|O(t) - \prod_i o_i \cdot I\| &= \max\left\{\left|\lambda_\textup{max}(O(t)) -\prod_i o_i \right|,  \left|\lambda_\textup{min}(O(t)) -\prod_i o_i \right| \right\}\\
        &=\max\left\{\lambda_\textup{max}(O(t)) -\prod_i o_i ,  \prod_i o_i - \lambda_\textup{min}(O(t)) \right\}\,,
    \end{align} where $\lambda_\textup{max}(O(t))$ and $\lambda_\textup{min}(O(t))$ denote the highest and lowest eigenvalue of $O(t)$ respectively. We have that \begin{equation}
        \lambda_\textup{max}(O(t)) -\prod_i o_i \leq \prod_i (o_i + c_i e^{-\Delta_i t}) - \prod_i o_i\,.
    \end{equation} Note that $\prod_i (o_i + c_i e^{-\Delta_i t}) - \prod_i o_i \leq \epsilon$ is equivalent to $\prod_i (1 + \frac{c_i}{o_i} e^{-\Delta_i t}) - 1 \leq \prod_i o_i^{-1} \cdot \epsilon$, and since $\prod_i o_i^{-1} \cdot \epsilon \geq o_\textup{max}^{-k(n)} \cdot \epsilon$ and $\prod_i (1 + \frac{c_i}{o_i} e^{-\Delta_i t}) \leq (1 + \frac{c_\textup{max}}{o_\textup{min}} e^{-\Delta_\textup{min} t})^{k(n)}$, setting \begin{equation}
        \left(1 + \frac{c_\textup{max}}{o_\textup{min}} e^{-\Delta_\textup{min} t}\right)^{k(n)} - 1 \overset{\text{set}}{\leq}  o_\textup{max}^{-k(n)} \cdot \epsilon
    \end{equation} implies the desired inequality. Solving this yields \begin{equation}
        t \geq t_1 \coloneqq \frac{1}{\Delta_\textup{min}} \log\left( \frac{c_\textup{max}}{o_\textup{min}} \frac{1}{\left(1+\epsilon \cdot o_\textup{max}^{-k(n)}\right)^{1/k(n)}-1} \right)\,.
    \end{equation} 
    Now consider $t \geq t_0 \coloneqq \frac{1}{\Delta_\textup{min}} \log\left( \frac{c_\textup{max}}{o_\textup{min}}\right)$ so that $o_i - c_i e^{-\Delta_i t} \geq 0$ for all $i$, and hence $\lambda_\textup{min}(O(t)) \geq \prod_i (o_i - c_i e^{-\Delta_i t})$ (here, $c_\textup{max}$ can be always assumed to be greater or equal to $1$, and so $t_0 >0$). Then for all $t\geq t_0$, we find that \begin{equation}
        \prod_i o_i - \lambda_\textup{min}(O(t)) \leq  \prod_i o_i - \prod_i (o_i - c_i e^{-\Delta_i t})\,.
    \end{equation} Similarly to before, we can show that for any \begin{equation}
        t \geq t_2 \coloneqq \frac{1}{\Delta_\textup{min}} \log\left( \frac{c_\textup{max}}{o_\textup{min}} \frac{1}{1-(1-\epsilon \cdot o_\textup{max}^{-k(n)})^{1/k(n)}} \right) \geq t_0\,,
    \end{equation} we have $\prod_i o_i - \prod_i (o_i - c_i e^{-\Delta_i t}) \leq \epsilon$. Note that for $\epsilon >o_\textup{max}^{k(n)}$, the desired inequality becomes vacuously true, and then $t_2 = t_0$ surely suffices. 
    Finally, observe that for any $\epsilon < o_\textup{max}^{k(n)}$, we have $t_1 \geq t_2$, which follows from the inequality $(1+x)^{1/m} + (1-x)^{1/m} \leq 2$ (which holds for all $x \in [-1,1]$ and $m \in \mathbb{N}$) by taking $x = \epsilon\cdot o_\textup{max}^{-k(n)} \in (0,1]$ and $m=k(n)$.
    Altogether, this means that for all $t \geq \max\{t_0,t_1\}$, $\|O(t) - \prod_i o_i \cdot I\|\leq \epsilon$, yielding the following bound on the mixing time: \begin{equation}
        t_\textup{mix}^{(O)}(\epsilon) \leq \max\{t_0,t_1\} = \frac{1}{\Delta_\textup{min}} \log\left( \frac{c_\textup{max}}{o_\textup{min}} \cdot \max\left\{\frac{1}{\left(1+\epsilon\cdot o_\textup{max}^{-k(n)} \right)^{1/k(n)}-1},\,1 \right\}\right)\,.
    \end{equation}
    Note that, if $k(n)$ is a bounded function of $n$, then this clearly gives an upper bound on the mixing time uniform in $n$. Now consider $k(n)\to \infty$ as $n\to\infty$, then for a fixed $\epsilon$ and $n \to \infty$, this behaves like $\max\{t_0,t_1\}\sim \frac{1}{\Delta_\textup{min}} \log\left( \frac{c_\textup{max}}{o_\textup{min}} \cdot \max\left\{\frac{o_\textup{max}}{1-o_\textup{max}},1 \right\} \right)$, while for a fixed $n$ and $\epsilon \to 0$ like $\max\{t_0,t_1\}\sim \frac{1}{\Delta_\textup{min}} \log\left( \frac{c_\textup{max}}{o_\textup{min}} \frac{k(n)}{o_\textup{max}^{-k(n)}}\cdot\frac{1}{\epsilon} \right)$. In either case, we find a mixing time which is upper bounded uniformly in $n$. Interestingly enough, as the asymptotic for a fixed $\epsilon$ and $k(n) \to \infty$ happens to be independent of $\epsilon$, then we might see a mixing time which after some initial increasing behaviour starts decreasing towards this value as you increase the system size $n$, or is decreasing immediately from $n=1$. Observe that as $\epsilon \to 0$, the maximum of $\max\{t_0,t_1\}$ occurs at $k(n) \sim \frac{1}{-\log(o_\textup{max})}$ (which can be smaller than $1$), where we have $\max\{t_0,t_1\} \sim \frac{1}{\Delta_\textup{min}} \left( \log(\frac{c_\textup{max}}{o_\textup{min}} \cdot \frac{1}{-\log(o_\textup{max})}\cdot\frac{1}{\epsilon})-1\right)$.
\end{proof}

\begin{remark}\label{remark: constant mixing time with increasing accuracy}
    This bound on the mixing time for extensive product observables further means that we end up with a constant mixing time even for a sub-exponentially decreasing error $\epsilon = \epsilon(n)$. In particular, if $k(n) = \Theta(n)$, then for any $\epsilon = 1/\operatorname{poly}(n)$ we still get a a system-size-independent bound on the mixing time.
\end{remark}

\begin{corollary}
    Consider a separable Lindbladian $\mathcal{L}$ with a full-rank steady state. Any observable $O = \sum_{i \in S} O_i$, with $O_i$ being product observables with extensive support $\operatorname{supp}(O_i) = \Theta(n)$, $|S| = \operatorname{poly}(n)$, and $\|O\| = \Omega(1/\operatorname{poly(n)})$, mixes in a constant time.
\end{corollary}
\begin{proof}
    By the assumptions, we have that $\frac{\|O\|}{\sum_i \|O_i\|} = 1/\operatorname{poly}(n)$. Hence by Remark \ref{remark: constant mixing time with increasing accuracy}, we have that $t_\textup{mix}^{(O_i)}\left(\epsilon \cdot \frac{\|O\|}{\sum_i \|O_i\|}\right)$ is upper bounded by a constant for each $i$, and so by Lemma \ref{lemma: mixing time of a sum} we have that \begin{equation}
        t_\textup{mix}^{(O)}(\epsilon) \leq \sup_i t_\textup{mix}^{(O_i)}\left(\epsilon \cdot \frac{\|O\|}{\sum_i \|O_i\|}\right)\,,
    \end{equation} which is system-size-independent. 
\end{proof}

\begin{example}[Toy example of an observable with logarithmic mixing time]\label{example: observable mixing in log time}
    Consider the (1-local) depolarising channel $e^{t\mathcal{L}}$ generated $\mathcal{L} = \sum_i \mathcal{L}_i$ acting on $n$ qubits, where $\mathcal{L}_i[O] = \frac{1}{4}(X_i O X_i + Y_i O Y_i + Z_i O Z_i - 3\cdot O)$. Fix a desired accuracy $\epsilon \in (0,1)$. Then the majority observable \begin{equation}
        O = \sum_{x \in B_n} \operatorname{sgn}(n-2|x|) \cdot |x\rangle\langle x|\,,
    \end{equation} where $|x|$ denotes the Hamming weight of the bit string $x$ and $\operatorname{sgn}$ denotes the sign function, will mix in a logarithmic time lower bounded as \begin{equation}
        t_\textup{mix}^{(O)}(\epsilon) \geq \frac{1}{2} \log\left( \frac{1}{1-(1-\epsilon)^{1/\lceil n/2\rceil}}\right) \sim \frac{1}{2} \log\left(\frac{n}{2\epsilon} \right)\,.
    \end{equation}
\end{example}
\begin{proof}
    Observe that $O$ is a diagonal observable, with diagonal elements $\pm 1$ depending whether or not the bit string $x$ corresponding to the basis state $|x\rangle$ contains more zeros or ones respectively. In case of a possible tie, there would also be zeros on the diagonal, however, we can for simplicity think of only odd $n$'s so that ties aren't possible, making $O$ also unitary (although the rest of the proof will be adapted for all $n \in \mathbb{N}$). Hence we immediately get that $\|O\| = 1$, and also that $O(t)$ converges to $0 = \Tr\left(O \cdot \frac{I}{2^n}\right) \propto \Tr(O)$. As such, the mixing condition simplifies to $\|O(t)\| \leq \epsilon$.

    Observe that $\mathcal{L}_i [I_i] = 0$ and $\mathcal{L}_i [P_i] = -P_i$ for a Pauli $P_i$, hence get that $e^{t\mathcal{L}_i}[|x_i\rangle\langle x_i|] = \frac{I_i}{2} +(-1)^{x_i} \frac{e^{-t}}{2} Z_i$. In particular, we get that \begin{equation}
        O(t) = \sum_{x_1,\dots,x_n = 0}^1 \operatorname{sgn}(n-2\cdot(x_1+\dots+x_n)) \cdot \left(\frac{I}{2} +(-1)^{x_1} \frac{e^{-t}}{2} Z \right) \otimes \dots \otimes \left(\frac{I}{2} +(-1)^{x_n} \frac{e^{-t}}{2} Z \right)\,.
    \end{equation} Since this is a sum of diagonal terms, we see that $O(t)$ remains being diagonal throughout its evolution. Its diagonal elements are then \begin{align}\hspace{-1cm}
        \langle y|O(t)|y\rangle &= \sum_{x_1,\dots,x_n = 0}^1 \operatorname{sgn}(n-2\cdot(x_1+\dots+x_n)) \cdot \langle y_1 | \left(\frac{I}{2} +(-1)^{x_1} \frac{e^{-t}}{2} Z \right)|y_1\rangle  \cdots  \langle y_n|\left(\frac{I}{2} +(-1)^{x_n} \frac{e^{-t}}{2} Z \right)|y_n\rangle\hspace{-1cm} \\
        &= \left( \frac{1+e^{-t}}{2}\right)^n \cdot \sum_{x \in B_n} \operatorname{sgn}(n-2|x|) \cdot \left( \frac{1-e^{-t}}{1+e^{-t}}\right)^{|x-y|}\\
        &= \left( \frac{1+e^{-t}}{2}\right)^n \cdot \sum_{k=0}^n \left( \frac{1-e^{-t}}{1+e^{-t}}\right)^{k}  \cdot \sum_{\substack{x \in B_n\\ |x-y| = k}} \operatorname{sgn}(n-2|x|)\,,
    \end{align} where $|x-y|$ denotes the Hamming distance between $x$ and $y$. Due to the symmetry of this expression, observe that $\langle y|O(t)|y\rangle = - \langle \overline{y}|O(t)|\overline{y}\rangle$, where $\overline{y}$ denotes the bitwise complement to $y$. Further, as $q=\frac{1-e^{-t}}{1+e^{-t}} \in [0,1)$, so that the most weight lies in small $k$'s, it will follow that the largest diagonal element occurs at $y=0$. To see why, first split the sum over $k$ into pairs $k=l$ and $k=n-l$ as follows: \begin{align}
        \sum_{k=0}^n q^{k}  \cdot \sum_{\substack{x \in B_n\\ |x-y| = k}} \operatorname{sgn}(n-2|x|) &= \sum_{l=0}^{\lfloor n/2\rfloor} \left( q^l \cdot \sum_{\substack{x \in B_n\\ |x-y| = l}} \operatorname{sgn}(n-2|x|) + q^{n-l} \cdot \sum_{\substack{x \in B_n\\ |x-y| = n-l}} \operatorname{sgn}(n-2|x|)\right)\\
        &= \sum_{l=0}^{\lfloor n/2\rfloor} \left( q^l \cdot \sum_{\substack{x \in B_n\\ |x-y| = l}} \operatorname{sgn}(n-2|x|) + q^{n-l} \cdot \sum_{\substack{x \in B_n\\ |x-\overline{y}| = l}} \operatorname{sgn}(n-2|x|)\right)\\
        &=  \sum_{l=0}^{\lfloor n/2\rfloor} \left( (q^l-q^{n-l}) \cdot \sum_{\substack{x \in B_n\\ |x-y| = l}} \operatorname{sgn}(n-2|x|) \right)\,,
    \end{align} where we have temporarily assumed $n$ to be odd. As $q^l - q^{n-l} >0$, each of these terms will be maximised when $ \sum_{\substack{x \in B_n\\ |x-y| = l}} \operatorname{sgn}(n-2|x|)$ is maximised, which in turn achieves its overall maximum whenever $|y| < \frac{n}{2}-l$. As $l$ goes up to $\lfloor n/2 \rfloor$, taking $y=0$ maximises all these terms individually (and is in fact the only option that does that), and hence maximises the whole expression. For $n$ being even, we would be further left with a middle term $q^{n/2} \cdot  \sum_{\substack{x \in B_n\\ |x-y| = n/2}} \operatorname{sgn}(n-2|x|)$, which is actually equal to $0$ for all $y$, as $\sum_{\substack{x \in B_n\\ |x-y| = n/2}} \operatorname{sgn}(n-2|x|) = \sum_{\substack{x \in B_n\\ |x-\overline{y}| = n/2}} \operatorname{sgn}(n-2|x|) = -\sum_{\substack{x \in B_n\\ |x-y| = n/2}} \operatorname{sgn}(n-2|x|)$.

    This gives us the spectral norm of $O(t)$ as \begin{align}
        \|O(t)\| &= \left( \frac{1+e^{-t}}{2}\right)^n \cdot \sum_{k=0}^n \left( \frac{1-e^{-t}}{1+e^{-t}}\right)^{k}  \cdot \sum_{\substack{x \in B_n\\ |x| = k}} \operatorname{sgn}(n-2|x|)\\
        &= \left( \frac{1+e^{-t}}{2}\right)^n \cdot \sum_{k=0}^n \left( \frac{1-e^{-t}}{1+e^{-t}}\right)^{k}  \cdot {n \choose k}\cdot \operatorname{sgn}(n-2k)\\
        &= \left( \frac{1+e^{-t}}{2}\right)^n \cdot \left( \sum_{k=0}^{\lfloor n/2 \rfloor} \left( \frac{1-e^{-t}}{1+e^{-t}}\right)^{k}  \cdot {n \choose k} - \sum_{k=\lceil n/2\rceil}^{n} \left( \frac{1-e^{-t}}{1+e^{-t}}\right)^{k}  \cdot {n \choose k}\right)\,,\label{eqn: example of log mixing time, spectral norm}
    \end{align} where, now, in case of even $n$, the middle term simply cancels out. Let's define $S(q) = \sum_{k=0}^{\lfloor n/2 \rfloor} q^k  \cdot {n \choose k}$ and assume $n$ to be odd. Then observe that $S(q) + q^n \cdot S(1/q) = (1+q)^n$, which follows from the binomial theorem together with the symmetry of binomial coefficients. Note that we can upper bound $S(1/q) \leq 2^{n-1} \cdot q^{-\lfloor n/2\rfloor}$, and hence we can lower bound $S(q) = (1+q)^n - q^n \cdot S(1/q) \geq (1+q)^n - 2^{n-1} \cdot q^{\lceil n/2\rceil}$. Using this inequality, we obtain a lower bound on the spectral norm of $O(t)$ as \begin{align}
         \|O(t)\| &=\left( \frac{1+e^{-t}}{2}\right)^n \cdot \left( 2\sum_{k=0}^{\lfloor n/2 \rfloor} \left( \frac{1-e^{-t}}{1+e^{-t}}\right)^{k}  \cdot {n \choose k} - \left(\frac{2}{1+e^{-t}}\right)^n\right)\\
         &\geq \left( \frac{1+e^{-t}}{2}\right)^n \cdot \left( \left(\frac{2}{1+e^{-t}}\right)^n - 2^n \cdot \left(\frac{1-e^{-t}}{1+e^{-t}}\right)^{\lceil n/2\rceil}\right)\\
         &= 1-(1+e^{-t})^{\lfloor n/2\rfloor} \cdot (1-e^{-t})^{\lceil n/2\rceil}\,,
    \end{align} which is further a strict inequality for any $t>0$. Now this inequality has been shown only for odd $n$, however, observe that the spectral norm actually stays the same when going from an odd $n$ to an even $n+1$, and hence we get that \begin{equation}
        \|O(t)\| \geq  1-(1+e^{-t})^{\lceil n/2\rceil-1} \cdot (1-e^{-t})^{\lceil n/2\rceil}
    \end{equation} for all $n \in \mathbb{N}$. To see that the spectral norm stays the same, observe that \begin{align}
        \hspace{-1cm}\|O(t)\|_{\text{even } n+1} &= \left( \frac{1+e^{-t}}{2}\right)^n \cdot \left( \frac{1+e^{-t}}{2}\right) \cdot \left( \sum_{k=0}^{\frac{n+1}{2}} \left( \frac{1-e^{-t}}{1+e^{-t}}\right)^{k}  \cdot {n+1 \choose k} - \sum_{k=\frac{n+1}{2}}^{n+1} \left( \frac{1-e^{-t}}{1+e^{-t}}\right)^{k}  \cdot {n+1 \choose k}\right)\\
        &= \left( \frac{1+e^{-t}}{2}\right)^n \cdot \left( \frac{1+e^{-t}}{2}\right) \cdot \left( \sum_{k=0}^{\frac{n-1}{2}} \left( \frac{1-e^{-t}}{1+e^{-t}}\right)^{k}  \cdot {n+1 \choose k} - \sum_{k=\frac{n+3}{2}}^{n+1} \left( \frac{1-e^{-t}}{1+e^{-t}}\right)^{k}  \cdot {n+1 \choose k}\right)\\
        &= \left( \frac{1+e^{-t}}{2}\right)^n \cdot \left( \frac{1+e^{-t}}{2}\right) \cdot \left( \sum_{k=0}^{\frac{n-1}{2}} \left( \frac{1-e^{-t}}{1+e^{-t}}\right)^{k}  \cdot {n \choose k} - \sum_{k=\frac{n+3}{2}}^{n} \left( \frac{1-e^{-t}}{1+e^{-t}}\right)^{k}  \cdot {n \choose k}\right.\\ &\left.\qquad +\sum_{k=1}^{\frac{n-1}{2}} \left( \frac{1-e^{-t}}{1+e^{-t}}\right)^{k}  \cdot {n \choose k-1} - \sum_{k=\frac{n+3}{2}}^{n} \left( \frac{1-e^{-t}}{1+e^{-t}}\right)^{k}  \cdot {n \choose k-1} - \left(\frac{1-e^{-t}}{1+e^{-t}}\right)^{n+1}\right)\\
        &= \left( \frac{1+e^{-t}}{2}\right)^n \cdot \left( \frac{1+e^{-t}}{2}\right) \cdot \left( \frac{1-e^{-t}}{1+e^{-t}}+1\right) \cdot \left( \sum_{k=0}^{\frac{n-1}{2}} \left( \frac{1-e^{-t}}{1+e^{-t}}\right)^{k}  \cdot {n \choose k} - \sum_{k=\frac{n+1}{2}}^{n} \left( \frac{1-e^{-t}}{1+e^{-t}}\right)^{k}  \cdot {n \choose k}\right)\hspace{-1cm}\\
        &= \|O(t)\|_{\text{odd } n}\,.
    \end{align}
    (The spectral norm $\|O(t)\|$ is in fact a non-decreasing function of $n$, and strictly increasing on odd $n$'s for any $t>0$.)
    Finally, we can bound \begin{equation}
        1-(1+e^{-t})^{\lceil n/2\rceil -1} \cdot (1-e^{-t})^{\lceil n/2\rceil} \geq 1-(1+e^{-t})^{\lceil n/2\rceil} \cdot (1-e^{-t})^{\lceil n/2\rceil} = 1-(1-e^{-2t})^{\lceil n/2\rceil} \overset{\textup{set}}{>} \epsilon
    \end{equation} and set this greater to $\epsilon$, which then implies $
        t < t_1 \coloneqq \frac{1}{2} \log\left( \frac{1}{1-(1-\epsilon)^{1/\lceil n/2\rceil}} \right)$. This means that the mixing time of $O$ is lower bounded like \begin{equation}
            t_\textup{mix}^{(O)}(\epsilon) \geq t_1 = \frac{1}{2} \log\left( \frac{1}{1-(1-\epsilon)^{1/\lceil n/2\rceil}} \right)\,.
        \end{equation}
    By considering a fixed $\epsilon$ and taking $n \to \infty$, we get that $t_1 \sim \frac{1}{2} \log\left(\frac{n}{2\log(1/(1-\epsilon))}\right)$, which further behaves as $\sim \frac{1}{2} \log\left(\frac{n}{2\epsilon}\right)$ for $\epsilon \to 0$. Conversely, when first fixing $n$ and taking $\epsilon \to 0$, we get that $t_1 \sim \frac{1}{2} \log\left(\frac{\lceil n/2\rceil}{\epsilon}\right)$, which further behaves as $\sim \frac{1}{2} \log\left(\frac{n}{2\epsilon}\right)$ for $n \to \infty$. These matching asymptotics show that $t_1 \sim \frac{1}{2} \log\left(\frac{n}{2\epsilon}\right)$ independently of the direction of approach, and so they also apply for $\epsilon = \epsilon(n)$ depending on $n$.
\end{proof}

\pagebreak
\section{Fermionic systems}\label{sec: fermions}

In this section, we consider specifically the quantum Gibbs samplers introduced in \cite{ding2025efficient} in the free fermionic settings as studied in \cite{smid2025polynomial,smid2025rapid, tong2025fast}.

\begin{prop}\label{prop: constant fermionic quadratic observable mixing time}
    For a free fermionic system with a bounded single particle Hamiltonian, quadratic fermionic observables $O$ mix in a constant time bounded by \begin{equation}
        t_\textup{mix}^{(O)}(\epsilon ) \leq \frac{1}{2\Delta_0} \log\left(\frac{2}{\epsilon}\right)\,.
    \end{equation}
\end{prop}

\begin{proof}
    For convenience, write the Hamiltonian using Majorana fermions as $H_0 = \sum_{i,j} \omega_i h_{ij} \omega_j$, where $h$ is a hermitian and anti-symmetric $2n\times 2n$ matrix, and the Majorana fermions obey $\{\omega_i,\omega_j\} = 2\delta_{ij}$ together with $\omega_i^\dagger = \omega_i$. Bounded single particle Hamiltonian here means that $\|h\|=\mathcal{O}(1)$. Here we take the set of jump operators to be the single-site Majorana fermions, $\mathcal{A} = \{\omega_i\}_{i=1}^{2n}$.

    Now let's consider the time evolution of a quadratic fermionic observable $O = \sum_{i,j} \omega_i \Gamma_{ij} \omega_j$ in the Heisenberg picture, $O(t) = e^{t\mathcal{L}}[O]$, where $\Gamma$ is also a hermitian and anti-symmetric $2n\times 2n$ matrix. Since we're evolving a quadratic operator by a quadratic Lindbladian, we may expect it to remain quadratic throughout the evolution, and hence can write \begin{equation}
        O(t) = \sum_{i,j} \omega_i \Gamma_{ij}(t) \omega_j\,,
    \end{equation} together with the initial condition $\Gamma(0) = \Gamma$. The dynamics is then governed by \begin{align}
        \frac{\dd O(t)}{\dd t} = \sum_{i,j} \omega_i \dot \Gamma_{ij}(t) \omega_j = \mathcal{L}[O(t)] &= \sum_k \left( L_k^\dagger O(t) L_k - \frac{1}{2}\{L_k^\dagger L_k,O(t)\}\right)\\
        &=  \sum_{k,i,j}\left( L_k^\dagger \omega_i  \Gamma_{ij}(t) \omega_j L_k - \frac{1}{2}\{L_k^\dagger L_k,\omega_i  \Gamma_{ij}(t) \omega_j\}\right)\,.
    \end{align} Separating $\Gamma(t)$ into its symmetric and anti-symmetric parts, $\Gamma(t) = \Gamma^A(t) + \Gamma^S(t)$, and using \begin{equation} L_k = \sum_l [q(4h)e^{\beta h}]_{kl} \cdot \omega_l\,,\end{equation} we find that \begin{align} \hspace{-1cm}
        \sum_{i,j} \dot \Gamma^A_{ij}(t) \cdot \omega_i \omega_j + \Tr(\dot \Gamma^S(t)) \cdot I &= \sum_{i,j}\Big[ -4 e^{2\beta h} q(4h)^2 \Gamma^A(t) - 2 \Gamma^A(t) \sinh(2\beta h) q(4h)^2 + 2 q(4h)^2 \sinh(2\beta h) \Gamma^A(t) \hspace{-1cm}  \\ 
        &\qquad\qquad - 2 \Gamma^S(t) \sinh(2\beta h) q(4h)^2 + 2 q(4h)^2 \sinh(2\beta h) \Gamma^S(t)  \Big]_{ij} \cdot \omega_i \omega_j\,.
    \end{align} By splitting the RHS into its symmetric and anti-symmetric parts, we find that \begin{align}\hspace{-1cm}
        \sum_{i,j} \dot \Gamma^A_{ij}(t) \cdot \omega_i \omega_j + \Tr(\dot \Gamma^S(t)) \cdot I = \sum_{i,j}\left[ -2\cosh(2\beta h) q(4h)^2 \Gamma^A(t)-2\Gamma^A(t) \cosh(2\beta h) q(4h)^2 \right]_{ij} \cdot \omega_i \omega_j\hspace{1cm}\\  \hspace{1cm} +\Tr \Big( 2\sinh(2\beta h) q(4h)^2 \Gamma^S(t) - 2\Gamma^S(t)\sinh(2\beta h) q(4h)^2 - 2e^{2\beta h} q(4h)^2 \Gamma^A(t) +  2 \Gamma^A(t)e^{-2\beta h} q(4h)^2 \Big) \cdot I\,,\hspace{-1cm}
    \end{align} which can be separated into \begin{align}\hspace{-1cm}
        \dot\Gamma^A(t) &= -2\cosh(2\beta h) q(4h)^2 \Gamma^A(t)-2\Gamma^A(t) \cosh(2\beta h) q(4h)^2\,,\\
        \Tr(\dot \Gamma^S(t)) &=\Tr \Big( 2\sinh(2\beta h) q(4h)^2 \Gamma^S(t) - 2\Gamma^S(t)\sinh(2\beta h) q(4h)^2 - 2e^{2\beta h} q(4h)^2 \Gamma^A(t) +  2 \Gamma^A(t)e^{-2\beta h} q(4h)^2\Big)\,.\hspace{-1cm}
    \end{align} Using the initial condition $\Gamma(0) = \Gamma$, we can solve the equation for $\Gamma^A(t)$ with \begin{align}
        \Gamma^A(t) = e^{-2\cosh(2\beta h) q(4h)^2 \cdot t} \cdot \Gamma \cdot e^{-2\cosh(2\beta h) q(4h)^2 \cdot t} \,,
    \end{align} which leaves us with \begin{align}
        &\hspace{-1.2cm}\Tr(\dot \Gamma^S(t)) =\Tr \Big( 2\sinh(2\beta h) q(4h)^2 \Gamma^S(t) - 2\Gamma^S(t)\sinh(2\beta h) q(4h)^2\\ &\hspace{-1.2cm}- 2e^{2\beta h} q(4h)^2 e^{-2\cosh(2\beta h) q(4h)^2 \cdot t} \cdot \Gamma \cdot e^{-2\cosh(2\beta h) q(4h)^2 \cdot t} +  2 e^{-2\cosh(2\beta h) q(4h)^2 \cdot t} \cdot \Gamma \cdot e^{-2\cosh(2\beta h) q(4h)^2 \cdot t}e^{-2\beta h} q(4h)^2\Big)\,.\hspace{-1cm}
    \end{align} The traces will be certainly equal when the matrices themselves are equal: \begin{align}
        \hspace{-1cm}\dot \Gamma^S(t) =2&\sinh(2\beta h) q(4h)^2 \Gamma^S(t) - 2\Gamma^S(t)\sinh(2\beta h) q(4h)^2- 2e^{2\beta h} q(4h)^2 e^{-2\cosh(2\beta h) q(4h)^2 \cdot t} \cdot \Gamma \cdot e^{-2\cosh(2\beta h) q(4h)^2 \cdot t}\hspace{-1cm}\\ &+  2 e^{-2\cosh(2\beta h) q(4h)^2 \cdot t} \cdot \Gamma \cdot e^{-2\cosh(2\beta h) q(4h)^2 \cdot t}e^{-2\beta h} q(4h)^2\,,
    \end{align} which can be formally solved by \begin{align}
        \hspace{-1cm}\Gamma^S(t) = e^{2\sinh(2\beta h) q(4h)^2 \cdot t} \cdot \int_0^t&\Big[ - 2e^{2\beta h} q(4h)^2 e^{-2\cosh(2\beta h) q(4h)^2 \cdot x} \cdot \Gamma \cdot e^{-2\cosh(2\beta h) q(4h)^2 \cdot x} \\ &+2 e^{-2\cosh(2\beta h) q(4h)^2 \cdot x} \cdot \Gamma \cdot e^{-2\cosh(2\beta h) q(4h)^2 \cdot x}e^{-2\beta h} q(4h)^2 \Big]\dd x \cdot e^{-2\sinh(2\beta h) q(4h)^2 \cdot t}\,.\hspace{-1cm}
    \end{align} Since we only need the trace $\Tr(\Gamma^S(t))$, we can greatly simplify this integration by cycling $\Gamma$ to the end: \begin{align}
        \Tr( \Gamma^S(t)) &= -4 \Tr\left( \sinh(2\beta h) q(4h)^2 \int_0^t e^{-4\cosh(2\beta h) q(4h)^2 \cdot x}\ \dd x \cdot \Gamma \right)\\
        &= \Tr\left( \tanh(2\beta h) (e^{-4\cosh(2\beta h) q(4h)^2 \cdot t}-1) \cdot \Gamma\right)\,,
    \end{align} which finally gives us the time-evolved form of $O$ as \begin{equation}
        O(t) = \boldsymbol{\omega}^T \cdot e^{-2\cosh(2\beta h) q(4h)^2 \cdot t} \cdot \Gamma \cdot e^{-2\cosh(2\beta h) q(4h)^2 \cdot t} \cdot \boldsymbol{\omega} + \Tr\left( \tanh(2\beta h) (e^{-4\cosh(2\beta h) q(4h)^2 \cdot t}-1) \cdot \Gamma\right) \cdot I\,,\label{eqn: time-evolved observable}
    \end{equation} where we've utilised bold matrix-vector notation for convenience.

    Finally, for any observable $O$, we can get the following bound on the difference of its expectation values: \begin{align}
        \left| \Tr(O \cdot (\rho(t) - \sigma_\beta) )\right| = \left| \Tr (O
        (t) \cdot (\rho_0 - \sigma_\beta)) \right| &= \left| \Tr \left( \left(O
        (t)-\Tr(O(t)) \cdot \frac{I}{2^n} \right)\cdot (\rho_0 - \sigma_\beta)\right) \right| \\ 
        &\leq \left\| O
        (t)-\Tr(O(t)) \cdot \frac{I}{2^n} \right\| \cdot \|\rho_0 - \sigma_\beta\|_{\Tr} \\
        &\leq 2 \left\| O
        (t)-\Tr(O(t)) \cdot \frac{I}{2^n} \right\|\,,
    \end{align} and hence using the form \eqref{eqn: time-evolved observable}, we find that \begin{align}
        \left| \Tr(O \cdot (\rho(t) - \sigma_\beta) )\right| &\leq 2 \left\| \boldsymbol{\omega}^T \cdot e^{-2\cosh(2\beta h) q(4h)^2 \cdot t} \cdot \Gamma \cdot e^{-2\cosh(2\beta h) q(4h)^2 \cdot t} \cdot \boldsymbol{\omega} \right\|\\
        &= 2 \left\| e^{-2\cosh(2\beta h) q(4h)^2 \cdot t} \cdot \Gamma \cdot e^{-2\cosh(2\beta h) q(4h)^2 \cdot t} \right\|_{\Tr}\\
        &\leq 2 \left\|e^{-2\cdot \cosh(2\beta h) q(4h)^2 \cdot t}\right\|^2 \cdot \|\Gamma \|_{\Tr}\\
        &= 2 e^{-4\cdot \min_i q(4\epsilon_i)^2\cosh(2\beta \epsilon_i) \cdot t} \cdot \| O \|\,,
    \end{align} where $\epsilon_i \in \operatorname{spec}(h)$. Setting this smaller to $\epsilon \| O\|$ will imply the bound on the observable mixing time as per Definition \ref{def:mixing time of observables}. This will be satisfied whenever \begin{equation}
        t \geq \frac{1}{4\cdot \min_i q(4\epsilon_i)^2\cosh(2\beta \epsilon_i)} \cdot \log\left(\frac{2}{\epsilon}\right) \eqqcolon \frac{1}{2\Delta_0}\cdot \log\left(\frac{2}{\epsilon}\right)\,,
    \end{equation} yielding the result of the proposition.
\end{proof}

\begin{remark}
    Adding a symmetric part to the coefficient matrix $h$ amounts to adding a scalar to the observable, $O \mapsto O'=O+\Tr(h)\cdot I$. Then we have \begin{equation}
        \left\| O'
        (t)-\Tr(O'(t)) \cdot \frac{I}{2^n} \right\| = \left\| O
        (t)-\Tr(O(t)) \cdot \frac{I}{2^n} \right\| \leq e^{-2\Delta_0 t } \|O\| \leq e^{-2\Delta_0 t } (\|O\|+|\Tr(h)|) = e^{-2\Delta_0 t }\|O'\|\,,
    \end{equation} where the last equality follows from the spectrum of $O$ being even. Hence adding a symmetric part to $h$ can only decrease the mixing time, and so the result of Proposition \ref{prop: constant fermionic quadratic observable mixing time} holds for arbitrary hermitian $h$ with $\|h\|=\mathcal{O}(1)$. 
\end{remark}


\section{Bosonic systems}\label{sec: bosons}

In this section, we consider specifically the quantum Gibbs samplers introduced in \cite{ding2025efficient} in the free bosonic settings as studied in \cite{smid2025rapid}. For more general, non-Gaussian, settings of bosonic quantum Gibbs samplers, as well as details about their practical implementations, we refer the reader to \cite{becker2026computing, becker2026quantum, becker2026simulating}.

When we consider infinite-dimensional spaces, such as bosonic Fock spaces, we immediately run into complications with the Definition \ref{def:mixing time of observables}. Even simple operators appearing in these spaces are generally unbounded, and so this notion of mixing time is not well defined, as the spectral norms appearing therein are infinite. Instead, we need to introduce the following regularisation:

\begin{definition}\label{def:mixing time of bosonic observables}
    The mixing time of a bosonic observable $O$ under the evolution generated by $\mathcal{L}$ for an arbitrary initial state is 
    \begin{align}
        t_{\textup{mix}}^{(O)}(\epsilon) &= \inf\left\{ t\geq 0 \left|\lim_{M\to \infty}
    \frac{\|\Pi_M(e^{t\mathcal{L}}[O]-\Tr(O\cdot \sigma_\beta))\Pi_M\|}{\|\Pi_M O \Pi_M\|}\leq \epsilon\right.\right\}\,,
    \end{align} where $\Pi_M$ is the projector onto Fock states with at most $M$ total bosons.
\end{definition}

Note that this definition is equivalent to the original one for any bounded operator in a finite-dimensional space. The motivation here is that we consider the mixing time on any truncated finite-dimensional space, and define the mixing time of the observable as the limit of the dimension going to infinity\,---\,if it exists. Further, note that the total bosonic number is invariant under a unitary Bogoliubov transformation, $a_i \mapsto b_i = \sum_j U_{ij}\cdot a_j$, as $N^{(b)} = \sum_i b_i^\dagger b_i = \sum_{i,j,k}a_j^\dagger \overline{U}_{ij} U_{ik} a_k = \sum_{j,k} a_j^\dagger \delta_{jk}a_k = N^{(a)}$, which will be important later in the following proposition:

\begin{prop}\label{prop: constant quadratic bosonic observable mixing time}
    For a free bosonic system $H_0 = \sum_{i,j}a_i^\dagger h_{ij} a_j$ with a bounded single particle Hamiltonian, quadratic bosonic observables $O = \sum_{i,j}a_i^\dagger \Gamma_{ij}a_j + c\cdot I$, where $\Gamma$ is a real symmetric $n\times n$ matrix, mix in a constant time bounded by \begin{equation}
        t_\textup{mix}^{(O)}(\epsilon ) \leq \frac{1}{2\Delta_0} \log\left(\frac{1}{\epsilon}\right)\,.
    \end{equation}
\end{prop}
\begin{proof}
    Here we consider the Hamiltonian $H_0 = \sum_{i,j}a_i^\dagger h_{ij} a_j$, where $h$ is a real symmetric $n\times n$ matrix, and canonical bosonic operators obeying the CCR algebra $[a_i,a_j^\dagger] = \delta_{ij}$. By a bounded single particle Hamiltonian we now mean $\|h\| = \mathcal{O}(1)$ and $\|h^{-1}\| = \mathcal{O}(1)$. Following \cite{smid2025rapid}, we take the set of jump operators to be $\mathcal{A} = \{x_i,p_i\}_{i=1}^n$, where $x_i=a_i+a_i^\dagger$ and $p_i = -i(a_i-a_i^\dagger)$.

    First, we need to calculate the time evolution $O(t)$ of the quadratic bosonic observable $O$ in the Heisenberg picture, similarly to Proposition \ref{prop: constant fermionic quadratic observable mixing time}. We can again expect $O(t)$ to remain quadratic throughout the evolution, and hence can write \begin{equation}
        O(t) = \sum_{i,j}a_i^\dagger \Gamma_{ij}(t) a_j + c(t) \cdot I\,,
    \end{equation} together with the initial conditions $\Gamma(0) = \Gamma$ and $c(0) = c$. The dynamics is then governed by \begin{align}
        \frac{\dd O(t)}{\dd t} = \sum_{i,j}a_i^\dagger \dot\Gamma_{ij}(t) a_j + \dot c(t) \cdot I = \mathcal{L}[O(t)] &= \sum_\mu\left( L_\mu^\dagger O(t) L_\mu - \frac{1}{2}\{L_\mu^\dagger L_\mu,O(t)\}\right)\\
        &= \sum_\mu \sum_{i,j} \left( L_\mu^\dagger a_i^\dagger \Gamma_{ij}(t) a_j L_\mu - \frac{1}{2}\{L_\mu^\dagger L_\mu,a_i^\dagger \Gamma_{ij}(t) a_j\}\right)\,,
    \end{align} with \begin{equation}
        L_\mu = \begin{cases}
            \hat f(-h) \cdot \mathbf{a} + \hat f(h) \cdot \mathbf{a}^\dagger &\text{ for } 1\leq \mu \leq n\,,\\
            -i\hat f(-h) \cdot \mathbf{a} + i\hat f(h) \cdot \mathbf{a}^\dagger &\text{ for } n+1\leq \mu \leq 2n\,,\\
        \end{cases}
    \end{equation} where we've again utilised bold matrix-vector notation for convenience.
    In the bosonic setting, we can expect $\Gamma(t)$ to remain symmetric throughout the evolution. After some simplifications, we will find that \begin{align}
        \sum_{i,j}a_i^\dagger \dot\Gamma_{ij}(t) a_j + \dot c(t) \cdot I  &= \sum_{i,j} \left(-\Gamma(t) \hat f(-h)^2 -\hat f(-h)^2 \Gamma(t) + \hat f(h)^2 \Gamma(t) + \Gamma(t) \hat f(h)^2\right)_{ij}a_i^\dagger a_j\\ &\qquad + 2\cdot \Tr(\Gamma(t) \cdot \hat f(h)^2 ) \cdot I\,.
    \end{align} Using that $\hat f(\nu) = q(\nu) \cdot e^{-\beta \nu/4}$, where $q(\nu)$ is even, we can split this equation into \begin{align}
        \dot \Gamma(t) &= -2 \cdot \Gamma(t) \cdot q(h)^2 \sinh(\beta h/2) - 2 \cdot q(h)^2 \sinh(\beta h/2) \cdot \Gamma(t)\,,\\
        \dot c(t) &= 2 \cdot \Tr(q(h)^2 e^{-\beta h/2} \cdot \Gamma(t))\,. 
    \end{align} Together with the initial condition $\Gamma(0) = \Gamma$, we can solve the first equation by \begin{align}
        \Gamma(t) = e^{-2\cdot q(h)^2 \sinh(\beta h/2) \cdot t} \cdot \Gamma \cdot e^{-2\cdot q(h)^2 \sinh(\beta h/2) \cdot t}\,,
    \end{align} which then leads to the solution of the second equation as \begin{align}
        c(t) &= c+ 2 \cdot \Tr\left(q(h)^2 e^{-\beta h/2} \cdot \int_0^te^{-4\cdot q(h)^2 \sinh(\beta h/2) \cdot x}\dd x \cdot \Gamma\right)\\
        &= c+ \Tr\left(\frac{1}{e^{\beta h}-1} \cdot \left(1-e^{-4\cdot q(h)^2 \sinh(\beta h/2) \cdot t}\right) \cdot \Gamma\right)\,.
    \end{align} This finally gives us the time-evolved form of $O$ as \begin{equation}\hspace{-1cm}
        O(t) = \mathbf{a}^\dagger \cdot e^{-2\cdot q(h)^2 \sinh(\beta h/2) \cdot t} \cdot \Gamma \cdot e^{-2\cdot q(h)^2 \sinh(\beta h/2) \cdot t} \cdot \mathbf{a} + \Tr\left(\frac{1}{e^{\beta h}-1} \cdot \left(1-e^{-4\cdot q(h)^2 \sinh(\beta h/2) \cdot t}\right) \cdot \Gamma\right) \cdot I + c\cdot I\,.\hspace{-1cm}
    \end{equation}

    Now, we wish to evaluate the mixing time of $O$ as per Definition \ref{def:mixing time of bosonic observables}. In order to do that, we need to consider the effect of truncating the Fock space to states with at most $M$ bosons. Observe that the spectrum of $\Pi_M \cdot (O(t) - \Tr(O\cdot \sigma_\beta) \cdot I)\cdot \Pi_M$ is then \begin{equation}\hspace{-0.7cm}
        \operatorname{spec}(\Pi_M \cdot (O(t) - \Tr(O\cdot \sigma_\beta) \cdot I)\cdot \Pi_M) = \left\{ \left.\sum_i \lambda_i \cdot n_i - \Tr\left(\frac{1}{e^{\beta h}-1} \cdot e^{-4\cdot q(h)^2 \sinh(\beta h/2) \cdot t} \cdot \Gamma\right)\right| n_i \in \mathbb{N}_0, \sum_i n_i \leq M\right\}\,,\hspace{-1cm}
    \end{equation} where $\lambda_i \in \operatorname{spec}\left(e^{-2\cdot q(h)^2 \sinh(\beta h/2) \cdot t} \cdot \Gamma \cdot e^{-2\cdot q(h)^2 \sinh(\beta h/2) \cdot t}\right)$. For this result, it was important that the total number of bosons is invariant under a unitary Bogoliubov transformation. The spectral norm for sufficiently large $M$ is then simply \begin{align}
        \left\|\Pi_M \cdot (O(t) - \Tr(O\cdot \sigma_\beta) \cdot I)\cdot \Pi_M\right\| &= M \cdot \left\|e^{-2\cdot q(h)^2 \sinh(\beta h/2) \cdot t} \cdot \Gamma \cdot e^{-2\cdot q(h)^2 \sinh(\beta h/2) \cdot t} \right\|\\ &\qquad \mp \Tr\left(\frac{1}{e^{\beta h}-1} \cdot e^{-4\cdot q(h)^2 \sinh(\beta h/2) \cdot t} \cdot \Gamma\right)\,,
    \end{align} where the sign depends on the sign of the $\lambda_i$ corresponding to the spectral norm $\|e^{-2\cdot q(h)^2 \sinh(\beta h/2) \cdot t} \cdot \Gamma \cdot e^{-2\cdot q(h)^2 \sinh(\beta h/2) \cdot t}\|$. Similarly, we have that $\|\Pi_M O \Pi_M\| = M \cdot \|\Gamma\| \pm c$ for large enough $M$. Hence we finally get that \begin{align}
        \lim_{M\to \infty}\frac{\left\|\Pi_M \cdot (O(t) - \Tr(O\cdot \sigma_\beta) \cdot I)\cdot \Pi_M\right\|}{\|\Pi_M O \Pi_M\|} &= \frac{\left\|e^{-2\cdot q(h)^2 \sinh(\beta h/2) \cdot t} \cdot \Gamma \cdot e^{-2\cdot q(h)^2 \sinh(\beta h/2) \cdot t} \right\|}{\|\Gamma\|}\\
        &\leq e^{-4 \cdot\min_i q(\epsilon_i)^2 \sinh(\beta \epsilon_i/2) \cdot t}\\
        &\overset{\text{set}}{\leq} \epsilon\,,
    \end{align} where $\epsilon_i \in \operatorname{spec}(h)$. This is smaller than $\epsilon$ whenever \begin{equation}
        t \geq \frac{1}{4 \cdot\min_i q(\epsilon_i)^2 \sinh(\beta \epsilon_i/2)} \cdot\log\left(\frac{1}{\epsilon}\right) \eqqcolon \frac{1}{2\Delta_0} \cdot\log\left(\frac{1}{\epsilon}\right) \,,
    \end{equation} yielding the result of the proposition.
\end{proof}

\pagebreak
\section{Numerical simulations}\label{sec: numerics}

In order to complement our theoretical results as well as probe the differences in mixing behaviour in practice, we perform numerical simulations of quantum Gibbs samplers for small systems, comparing the exact mixing times for specific initial states with the energy-specific mixing time. For systems with up to $n=12$ sites (or, rather, encoded in up to $12$ qubits), we use techniques based on exact diagonalisation, while for larger systems we use tensor networks. We will show that the qualitatively different scaling significantly lowers the necessary evolution time compared to global state mixing even for small system sizes, achieving up to $2\times$ speed-up for only $n=12$ qubits in the case of a 1D transverse-field Ising model (Figure \ref{fig:TFIM_mixing}). The code for these simulations is available at \cite{Smid_GitHub_GibbsSampling}.

In the case of spin systems, we consider the transverse-field Ising model, described by the Hamiltonian \begin{equation}
    H_\text{Ising} = -J \sum_{\langle i,j\rangle} Z_i Z_j - h\sum_{i=1}^n X_i\,,\label{eqn: Hamiltonian TFIM}
\end{equation} where $\langle i,j\rangle$ represents nearest neighbours on a given lattice. For the Gibbs sampler, we choose single-site Pauli jumps together with the Gaussian filter function. Corollary \ref{cor: constant mixing weakly-interacting Gibbs} is hence applicable in the weakly-interacting regime $|J/h| \ll 1$. On Figure \ref{fig:TFIM_mixing}, we consider the TFIM on a 1D lattice at interaction strength $J/h = 0.5$ (note that this model undergoes a phase transition at $J/h = 1$), comparing the mixing times $t_\textup{mix}^{(\rho_0 = I/2^n)}(0.01)$ and $t_\textup{mix}^{(\rho_0 = |0\rangle\langle 0|)}(0.01)$ to the corresponding energy mixing times $t_\textup{mix}^{(H,\rho_0 = I/2^n)}(0.01)$ and $t_\textup{mix}^{(H,\rho_0 = |0\rangle\langle 0|)}(0.01)$ as well as the energy mixing time for arbitrary initial state $t_\textup{mix}^{(H)}(0.01)$. In both cases of the state mixing times, we see an increasing logarithmic scaling, while the energy-specific mixing time in all cases remains constant (or even decreases in the case $\rho_0 = |0\rangle\langle0|$) with increasing system size, confirming our predictions of Corollary \ref{cor: constant mixing weakly-interacting Gibbs}. Note that the energy-specific mixing times for systems with more than $12$ sites are obtained with tensor network simulations of the time-evolved Hamiltonians in the Heisenberg picture, whereas the state mixing times are simply extrapolated by fitting $a\log(b\cdot n)$ to the data up to $n=12$. Here, one can obtain the results for $t_\textup{mix}^{(H,\rho_0 = I/2^n)}$ and $t_\textup{mix}^{(H,\rho_0 = |0\rangle\langle 0|)}$ by simply evaluating the expectation values with the appropriate initial state, however, in order to obtain $t_\textup{mix}^{(H)}$, we needed to understand the spectral range of $H(t)$ using DMRG.

\begin{figure}[b]
    \centering
    \hspace{-1cm}\includegraphics[width=1.05\linewidth]{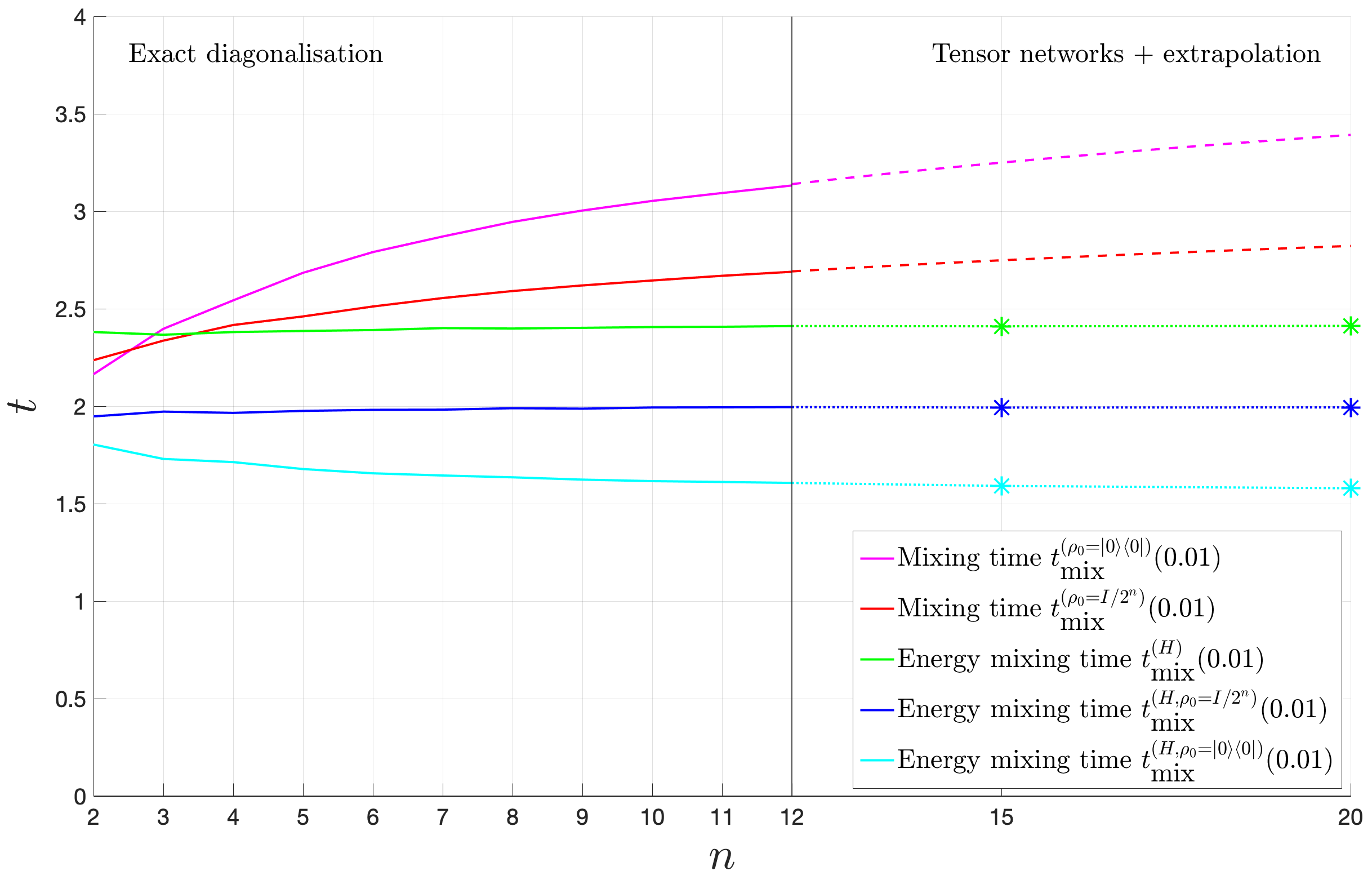}
    \caption{Comparing different mixing time notions for the 1D transverse-field Ising model with $J=0.5$, $h=1$, at $\beta =1$. Data up to $n=12$ are obtained via exact diagonalisation, while for $n=15$ and $20$ via tensor networks.}
    \label{fig:TFIM_mixing}
\end{figure}

\newpage
In fermionic settings, we look at the Fermi-Hubbard model given by
\begin{equation}
    H_\text{FH} = -t \sum_{\langle i,j\rangle, \sigma}( c_{i,\sigma}^\dagger c_{j,\sigma} + c_{j,\sigma}^\dagger c_{i,\sigma})+ U\sum_{i=1}^n N_{i,\uparrow}N_{i,\downarrow}\,,\label{eqn: Hamiltonian spinful FH}
\end{equation} as well as its spinless variant (also referred to as \textit{polarised}), given by
\begin{equation}
    H_\text{pFH} = -t \sum_{\langle i,j\rangle}( c_{i}^\dagger c_{j} + c_{j}^\dagger c_{i})+ U\sum_{\langle i,j\rangle}^n N_{i}N_{j}\,.\label{eqn: Hamiltonian spinless FH}
\end{equation}
For the Gibbs sampler, we similarly choose single-site Majorana jumps with the Gaussian filter.
The results for 1D lattices of the spinful version are shown on Figure \ref{fig:SpinfulFH_mixing}, while the spinless version is on Figure \ref{fig:SpinlessFH_mixing} (note that the spinless version of the model in 1D undergoes a phase transition at $U/t = 2$, while the spinful one doesn't have any critical points for $U/t > 0$). On both of these figures, we compare the mixing time $t_\textup{mix}^{(\rho_0 = I/2^n)}(0.01)$ to the corresponding energy mixing times $t_\textup{mix}^{(H,\rho_0 = I/2^n)}(0.01)$ as well as the energy mixing time for arbitrary initial state  $t_\textup{mix}^{(H)}(0.01)$. In the spinless case on Figure \ref{fig:SpinlessFH_mixing}, we can clearly see the energy-specific mixing times converging towards a constant value, while the state mixing time keeps on increasing logarithmically. For the spinful version on Figure \ref{fig:SpinfulFH_mixing}, this distinction is not so clear, as we can simulate only up to $n=6$ sites, which is not sufficient for convergence, but the energy-specific mixing times are increasing significantly slower than the state mixing time $t_\textup{mix}^{(\rho_0 = I/2^n)}(0.01)$.

The mixing time of an observable $O$ for an arbitrary initial state is defined via the contraction of its spectrum $\operatorname{spec}(O(t))$ around its steady value $\Tr(O\cdot \sigma)$. We explicitly demonstrate this on Figure \ref{fig:contractivity of spectrum} for the case of energy $O=H$ in the spinful Fermi-Hubbard model (Subfigure \ref{fig:SpinfulFH_spectral_range}) and the transverse-field Ising model (Subfigure \ref{fig:TFIM_spectral_range}) by plotting the highest and lowest eigenvalues of $H(t) = e^{t\mathcal{L}}[H]$ as they evolve in time, until the whole spectrum is contained around the steady value within an error $\epsilon = 0.01$ relative to the spectral norm $\|H\|$.

\begin{figure}[b]
    \centering
    \hspace{-1cm}\includegraphics[width=1.05\linewidth]{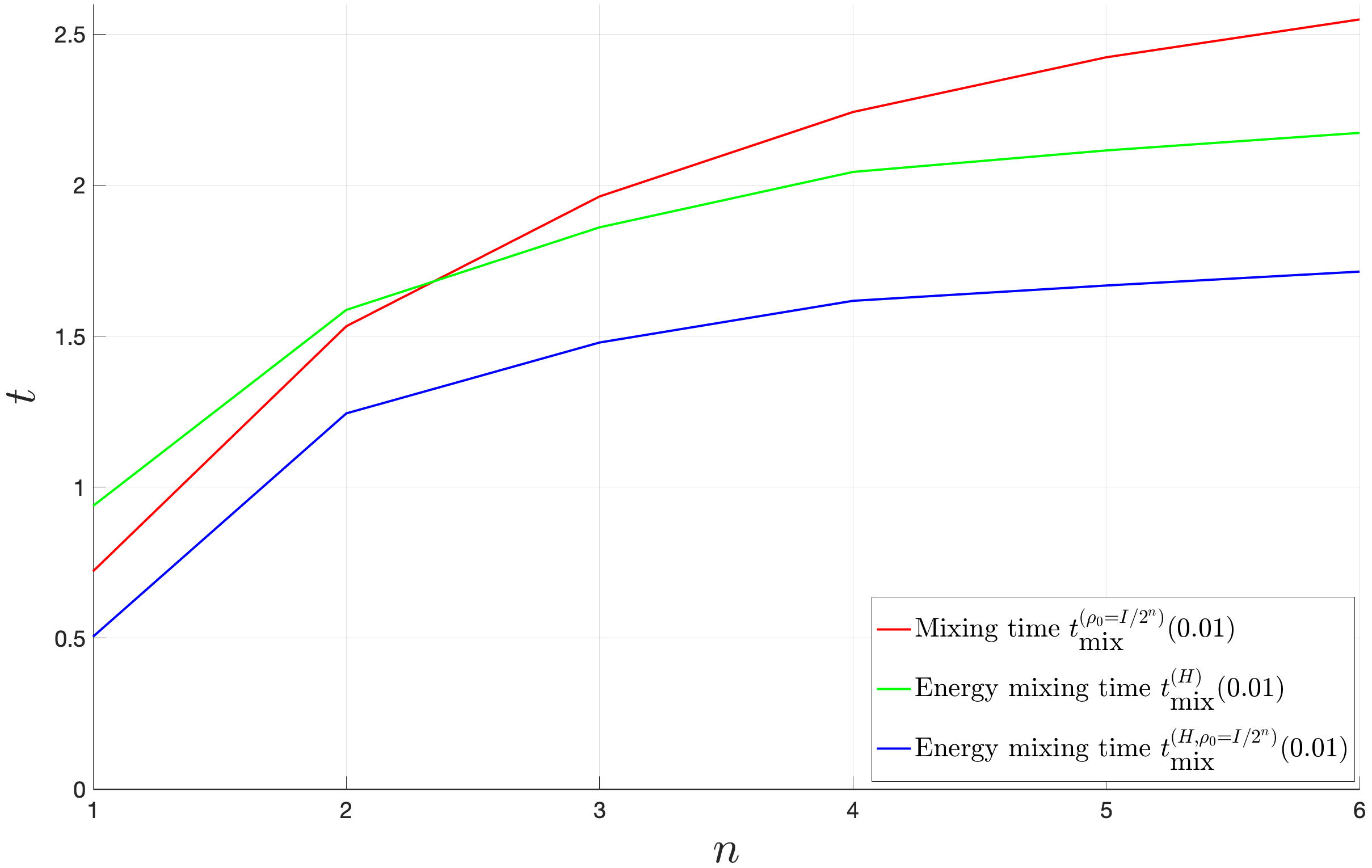}
    \caption{Comparing different mixing time notions for the spinful 1D Fermi-Hubbard model with $U=1$, $t=1$, at $\beta =1$.}
    \label{fig:SpinfulFH_mixing}
\end{figure}

\newpage

\begin{figure}[H]
    \centering
    \hspace{-1cm}\includegraphics[width=1.05\linewidth]{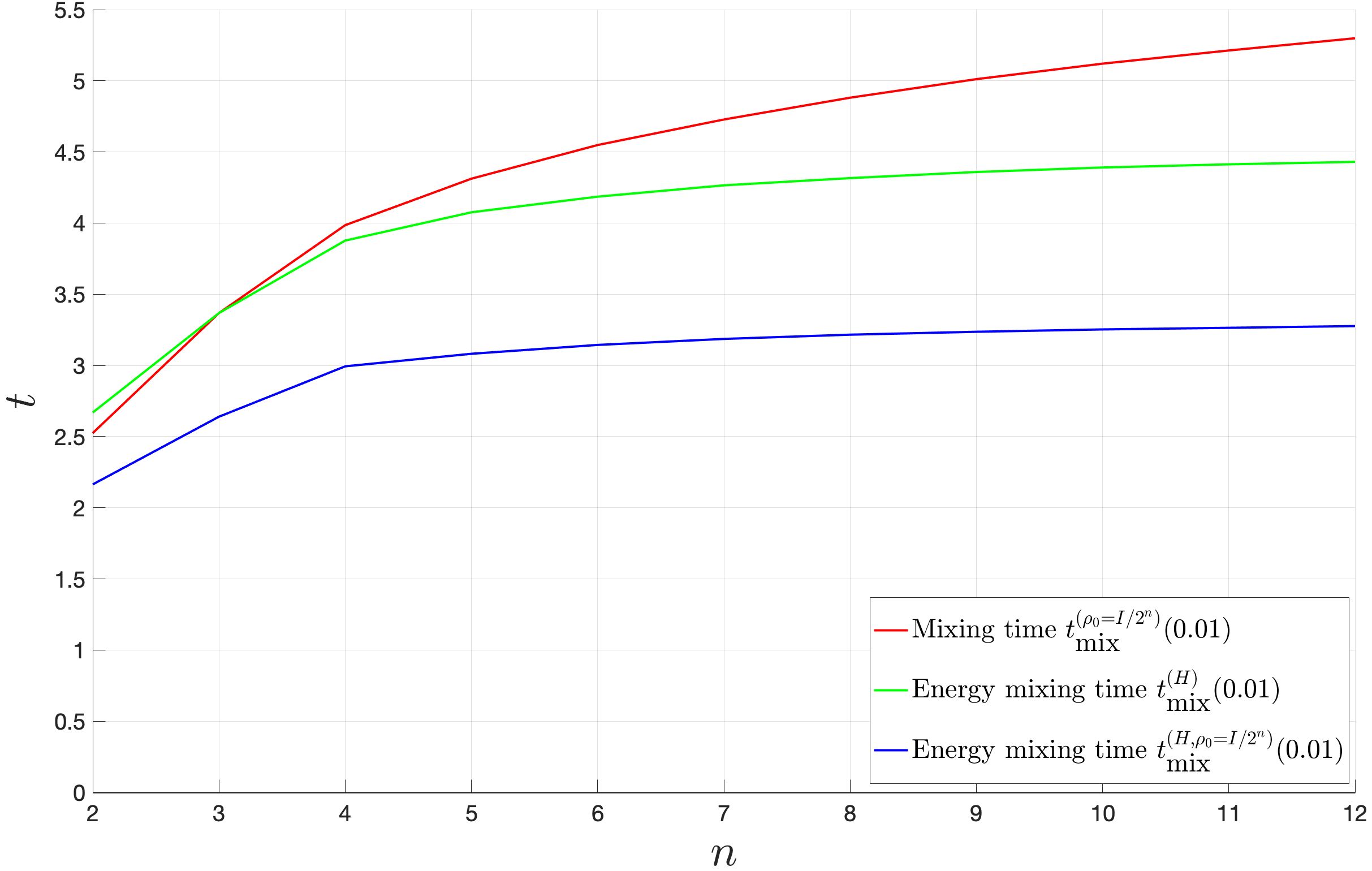}
    \caption{Comparing different mixing time notions for the spinless 1D Fermi-Hubbard model with $U=2$, $t=1$, at $\beta =1$.}
    \label{fig:SpinlessFH_mixing}
\end{figure}

\begin{figure}[H]
    \centering
    \begin{subfigure}[t]{0.5\linewidth}
        \centering
        \hspace{-1cm}\includegraphics[width=\linewidth]{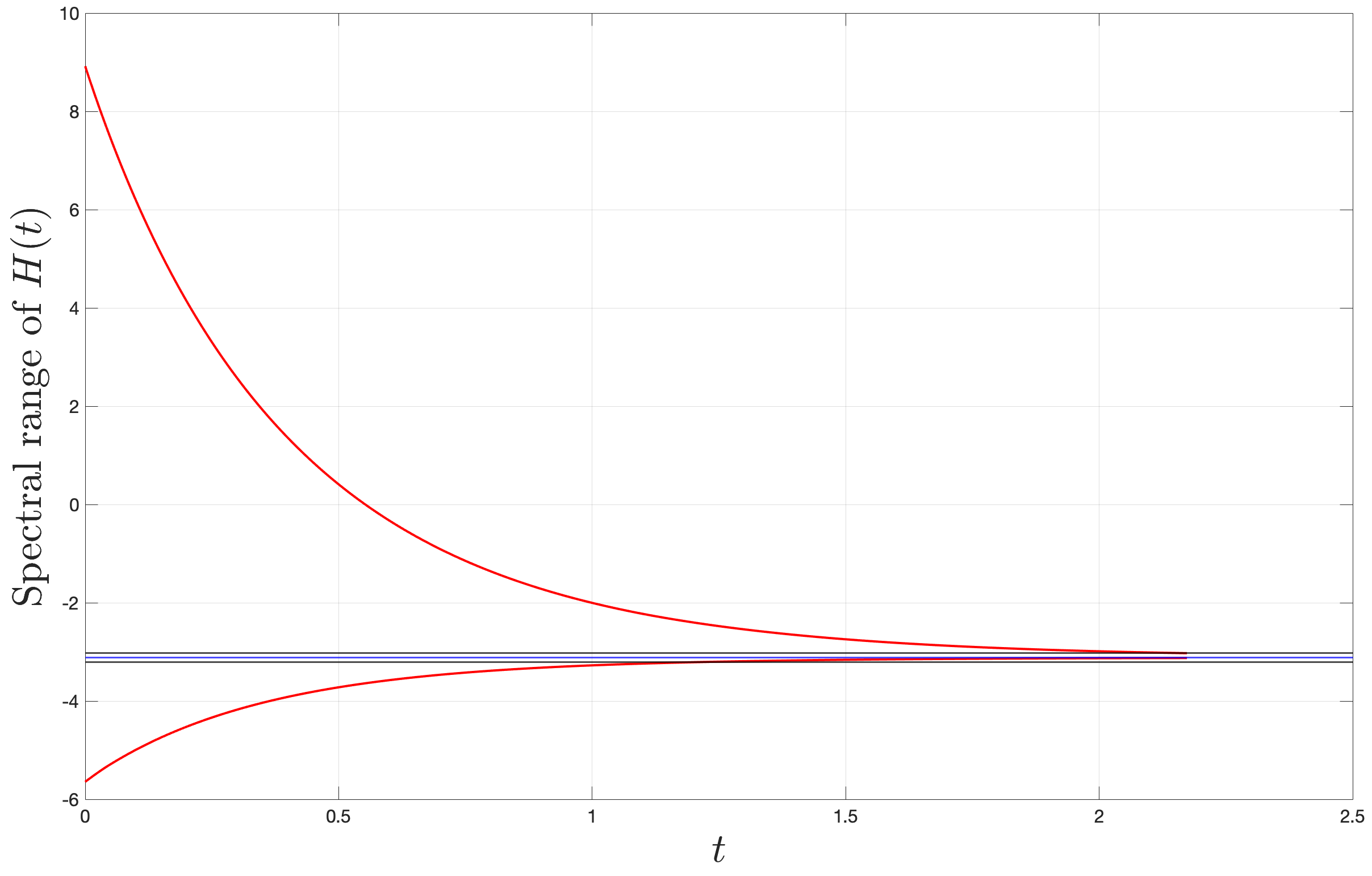}
    \caption{Spinful 1D FH model with $U=1$, $t=1$, at $\beta =1$, for $n=6$ sites.}\label{fig:SpinfulFH_spectral_range}
    \end{subfigure}%
    \begin{subfigure}[t]{0.5\linewidth}
        \centering
        \includegraphics[width=\linewidth]{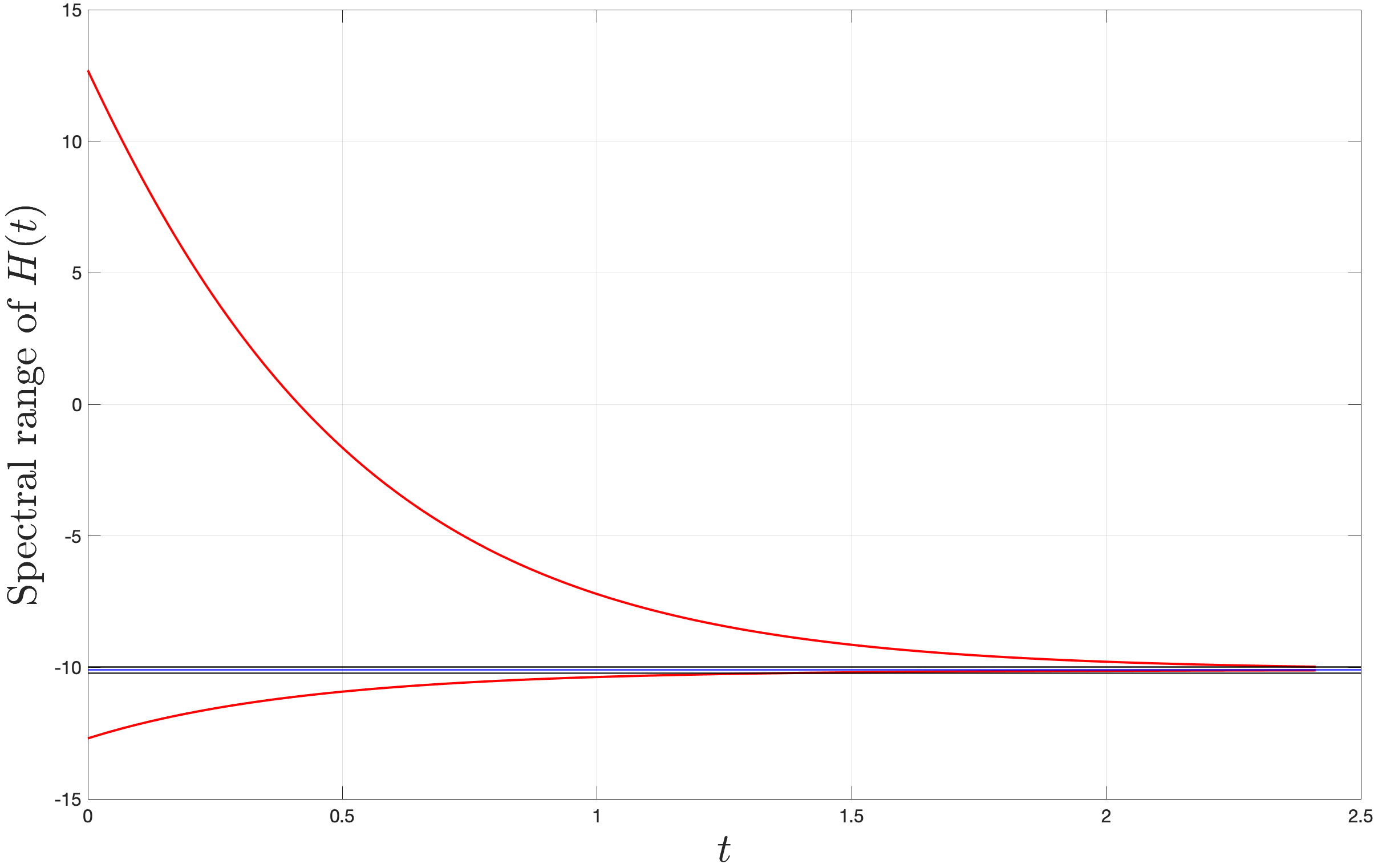}
    \caption{1D TFIM with $J=0.5$, $h=1$, at $\beta =1$, for $n=12$ sites.}\label{fig:TFIM_spectral_range}
    \end{subfigure}
    \caption{Contractivity of the spectral range of $H(t)$ around its expectation value $\Tr(H\cdot \sigma)$ in the steady state. The mixing time is defined by the whole spectrum being contained around the steady value within an error $\epsilon = 0.01$ relative to the spectral norm of $H$.}
    \label{fig:contractivity of spectrum}
\end{figure}

\acknowledgements
We acknowledge the use of Claude (series 5 models) when creating Figure \ref{fig:lightcone splitting}, when coming up with the observable in Example \ref{example: observable mixing in log time}, for debugging code, and for searching the literature on related works.
All authors acknowledge support from the EPSRC Grant number EP/W032643/1. MB acknowledges funding by the European Research Council (ERC Grant Agreement No.~948139) and the Excellence Cluster Matter and Light for Quantum Computing (ML4Q-2).

\newpage
\enlargethispage{0.5cm}
\bibliography{bibliography}
\bibliographystyle{alphaurl}

\end{document}